\documentclass[10pt]{article}
\usepackage{amsthm,amsmath,graphicx,hyperref,lmodern,amsfonts,amssymb,amsfonts,amsthm,multirow}
\usepackage{upgreek,booktabs,mathtools,enumerate,accents,bm,multicol,subcaption}

\usepackage[svgnames]{xcolor} 

\newcommand\diff{\mathop{}\!\mathrm{d}}

\hypersetup{
  colorlinks,
  linkcolor={FireBrick},
  citecolor={MidnightBlue}
}
\usepackage[backend=biber,style=apa]{biblatex}
\newtheorem{theorem}{Theorem}[section]
\newtheorem{corollary}{Corollary}[theorem]
\newtheorem{lemma}[theorem]{Lemma}
\theoremstyle{definition}
\newtheorem{definition}{Definition}[section]

\usepackage{Sweave}
\begin{document}
\title{Unbiased analytic non-parametric correlation estimators in the presence of ties}
\author{\href{mailto:ljrhurley@gmail.com}{Landon Hurley\footnote{Department of Projective Hedgehog Research, Alsace-Lorraine}}}
\maketitle

\begin{abstract}
An inner-product Hilbert space formulation is defined over a domain of all permutations with ties upon the extended real line. We demonstrate this work to resolve the common first and second order biases found in the pervasive Kendall and Spearman non-parametric correlation estimators, while presenting as unbiased minimum variance (Gauss-Markov) estimators. We conclude by showing upon finite samples that a strictly sub-Gaussian probability distribution is to be preferred for the Kemeny \(\tau_{\kappa}\) and \(\rho_{\kappa}\) estimators, allowing for the construction of expected Wald test statistics which are analytically consistent with the Gauss-Markov properties upon finite samples. 
\end{abstract}

Non-parametric correlations developed to extend the stable utility of the Pearson correlation coefficient to independently and identically sampled (\(i.i.d.\)), yet non-Gaussian, random variables. This problem requires identifying (1) an unbiased linear estimator of the distance between real vectors \(X^{n \times 1}\) and \(Y^{n \times 1}\), and a corresponding (2) limiting population variance of said distances. The first problem was partially addressed by the \textcite{spearman1904a} \(\rho\) and \textcite{kendall1938} \(\tau\) correlation coefficients. However, these two estimators upon bivariate random variable domains of finite length corresponding to the symmetric group of order \(n\), \(S_{n}\) show to be incomplete and therefore biased in many common empirical settings wherein ties arise. The domain restriction, or axiomatic population assumptions, which causes this is due to the loss of identification upon the observation of ties, which are mappings from observed data onto degenerate, or non-existence probabilities, resulting in non-measurable (and thus degenerate) axiomatic probability structures. 

Specifically, ties are surjective linear mappings wherein dependencies are observed and fail to be identified, such as when two separate observations \(i\) and \(i^{\prime}\) upon random variable \(X\) possess respective positive and negative linear relationships to their corresponding realisations upon \(Y\). Consider a general linear model with more than one covariate to reflect the linear combination of observations which result in, for some domain \(X^{n \times 2}\) with linear  projective regularity parameters \(\alpha\).  Such a problem is actually a natural result under the central limit theorem, wherein errors in regularity parameters \(\hat{\alpha} + \epsilon^{p \times 1}\to \alpha_{0}\) average out, rather than being homogeneously unidirectional (and thus monotonic), thus ensuring, almost surely, the observation of ties in linear models. 

In this manuscript, we mathematically derive two orthonormal estimator classes upon a common Hilbert space, which satisfy the Gauss-Markov theorem upon arbitrary or non-parametric random variables (both discrete and continuous), subject only to the assumption of an existing common uniform independent sampling upon a population of extended real scores. We also establish that these two classes are respective generalisations of the Kendall and Spearman correlations, and formally identify the explicit orthonormal geometric projective duality which exists between the two measure spaces, which was analytically presented by \textcite[p.~129]{kendall1948}.  We develop and present the asymptotic properties of the estimator families, and conclude with empirical simulations are performed which demonstrate the improvement in performance, especially upon discrete (e.g., point bi-serial and ordinal) data set. 

\section{Mathematical properties of the Kemeny estimator}

Both Kendall and Spearman defined a population domain of a sample upon a sequence of presumed rank transformations for bivariate random variable, subject to the strict expectation that no ties are observed. This condition, for continuous bivariate random variables, is most easily satisfied by the application of the birthday paradox, which remains solvable if and only if no ties are observed upon a probability mapping. This is because the event space, or domain, is strictly finite. In the event of ties upon the sample, wherein a surjective projection probability is degenerate, the expectation itself becomes degenerate, and thus violates the Borel-Cantelli lemma. This domain is the factorial space of the symmetric group of order \(n,\) written and counted as \(|S_{n}| = n!\). This space, in combination with an appropriate metric topology and a sampling selection process, allows for surjective mappings of individual bivariate relations to be scored using a distance function upon each marginal random variable \(X^{n \times 1},Y^{n \times 1}\in \mathbb{R}^{n \times 2}\), \(d: \mathbb{R}^{n \times 1} \times \mathbb{R}^{n \times 1} \to \mathbb{R}^{1 \times 1}\). This expresses the distance between vectors \((X^{n \times 1},Y^{n \times 1})\) as a real scalar upon the affine-linear function space and serves to solve the first problem uniquely, if the linear estimator is unbiased. If, for example, the domain is incomplete, it then follows that the mapping is non-unique, and consequently the linear distance is biased. 

A problem therefore arises in the examination of the Kendall \(\tau\) distance, as the construction of a Cayley graph is not a continuous Galois field and therefore is not easily amenable to the construction of an approximation function which directly, rather than implicitly, optimises upon it. In contrast is the Frobenius norm, which is a pragmatically universal function space upon which approximation functions are constructed, as easily verified by the matrix linear topology function space. The combination of axiomatic assumptions concerning the sampling which produces these random variable vectors, and the regular stochastic distribution of these values along with the linear function space (or metric space topology), results in an estimator. This problem is greatly simplified when said distance function is also provably a Hilbert function, as the almost sure uniqueness and point-wise convergence principles by the Central Limit Theorem ensures strong regularity in probabilistic convergence.

For Spearman's \(\rho\) the Frobenius norm (or the traditional Euclidean distance) is applied to said rankings, which are bijectively identified with a permutation upon \(S_{n}\) \parencite{diaconis1988}. Upon Kendall's \(\tau\) is employed the corresponding Kendall's \(\tau\)-distance function, which assesses the permutation distance between two arbitrary vectors mapped onto \(n!\), wherein the unit distance is defined as the swap, or re-ordering, of adjacent pairs of observations, \(i-1,i,i+1, i = 2,3,\ldots,n-1\) obtained upon the marginal random variable. For the common domain \(S_{n}\) from which the samples arise then exist two distinct distances, and thus Banach norm-spaces which each present the necessary probability mappings. 

This domain mapping denotes a permutation spanning basis, for which the relative positioning \(i\) allows for all \(n\) elements to be uniquely ordered, while remaining rotationally invariant over an affine-linear function space, which naturally includes monotonic transformations. Upon both estimators, affine-linear representations of the operator norms express relationships measured by the chosen distance function, which are both observed to be compact and totally bounded for any \(n \in \mathbb{N}^{+}.\)  For correlation coefficients, these affine-linear transformations of complete metric distance spaces (a pre-Hilbert, or Banach norm space is sufficient) represent the angles of similarity between the observed random variable pair, for which linearity between the pair of original scores is not a conditional expectation given by the presumption that either or both random variables be Gaussian distributed, thereby allowing for expansion beyond the Pearson point estimate upon the available sample.

Instead it is presumed that the rankings themselves, if not the original scores, are linearly comparable between the variables, and thus present a finite distance, almost surely. This presents an avenue of investigation concerning the linearity of both the scores and ranks, which results in a combinatorial sequence of four distinct possibilities, presented in Table~\ref{tab:combination_linearity}.
\begin{table}[!ht]
\centering
\footnotesize
\caption{Presentation of the observation of either a linear or non-linear manifold wrt the ranks and scores of upon a collection of random variables and their errors.}
\label{tab:combination_linearity}
\begin{tabular}{c|c|c}
\toprule
                            & Ranks & Ranks\\
\midrule
                    Scores    &   Linear & Non-linear\\
\midrule
               Linear      &   (1)     &   (2)  \\
\midrule
               Non-linear  &   (3)     &   (4) \\
\bottomrule
\end{tabular}
\end{table}

This model linearity upon the ranks and scores is the reasoning behind estimators being termed non-parametric, as regular linear structure upon both the sample and asymptotically is granted only by the identification of the observed probability density function, defined upon the \(\ell_{2}\)-Frobenius norm function space, and the error distribution of our approximation. Unsurprisingly, enforcing linearity by the incorporation of a Hilbert space allows for strong regularity in the error structure, thereby ensuring that a stable approximation is obtained (cf., constrained optimisation under the Karush-Kuhn-Tucker conditions, \cite{vapnik2013}). More colloquially, this can be expressed as the focus upon the linearity of the rank function space, rather than the potentially non-linear score manifold, the latter of which is only validly expected to be stably linear upon Gaussian random variables for finite samples. The problem for non-parametric function approximation (i.e., a general linear model of non-parametric yet linearly orderable errors) has been stymied by the inability to identify and measure observations of ties upon finite samples, particularly for the linear combination of multiple covariates, which often result in ties.  

The problem was originally resolved by invocation of the weak law of large numbers along with the central limit theorem, which ensures bilinear rank and score distributions which are identical and thus collinear. Without a proper linear and complete function space, finite sample estimators of the optimal ranking is however highly restricted. This is due to the inability to measure ties upon a probability measure space, which arise from both linear combination projections upon continuous and discrete model spaces, resulting in biased estimators which are only asymptotically uniquely convergent, and often analytically intractable. See, for example, bootstrapping and the empirical likelihood framework \parencite{owen2001} as a resolution to the linear characterisation of parameter uncertainty. We resolve the identification and optimality conditions with the Kemeny metric space, which we show here to be a Hilbert space from which arises, under very general conditions, a Gauss-Markov estimator which is strictly sub-Gaussian upon finite samples wrt \(n\). This paper is restricted to investigation of bivariate samples (and thus correlation coefficients). 

Condition 1 in Table~\ref{tab:combination_linearity} presents the standard Gaussian errors upon a target or dependent variable. By the asymptotic wrt \(n\) bilinearity condition of the weak law of large numbers, which establishes the central limit theorem upon the population, \(F_{x_{i}}^{-1} = x_{i},\) thereby denoting a collinear information system wrt both rank and score. However, this condition is only valid upon the population, as the estimator function space is restricted to the scores assessed by the Frobenius norm distance between two random variables, thereby resolving individual rankings to a stochastic error which tends to 0 only upon the population. However, for large Gaussian samples, this probabilistic argument is also observed to be stable, and therefore performs well in certain restricted undertakings \parencite{owen1988,owen2001}. We ignore Condition 2 as non-observable under the axiomatic definition that the scores and ranks arise upon a common population: if the scores are linear but the ranks are non-linear, it would imply non-linearly comparable scores upon the sample, indicating that a subset collection of scores may not be compared with their complementary group. Such scenarios arise in problems such as heteroscedasticity and mixture models, wherein linearly rankable scores would typically require a single common population linear function space to be validly applied, contradicting the nature of the prescribed data generating process; at worst however, the application of a misspecified linear function upon the scores would result in a biased or inefficient estimator. Condition 3 is substantively more common, as it reflects standard analytical problems which include the continuous and discrete exponential family of distributions, which are often analytically and numerically tractable only with large sample sizes \parencite{nelder1972}. The estimators presented in this work provide analytical solutions to this estimation problem, with a finite sample solution which is unique, tight, and maximally informative (compliant with the Gauss-Markov theorem) compared to established competitors. This work is also comparatively unique in that it provides linear rank models which are analytically solvable in both the first and second order approximations, thereby avoiding the often necessary bootstrapping procedure. Condition 4 is also axiomatically ignored, as it would consist of an attempted analysis by a single model of a mixture of different data generating functions, thereby producing biased and inefficient estimators due to the almost surely misspecified and unidentifiable model mis-specifications.

\subsection{Kemeny as a Hilbert space}
To resolve the problem of ties upon an arbitrary affine-linear function space, consider an extended domain of permutations with ties, which we define as \(\mathcal{M}_{n} = n^{n}-n \supset S_{n},\) the space of all permutations of length \(n \in \mathbb{N}^{+}\) explicitly allowing ties to be observed upon the marginal random variables. Explicitly excluded are the \(n\) conditional events upon the population in which are observed degenerate constant vectors. First, consider the relative cardinality of the two permutation spaces: \(n! \ll n^{n}-n,\) which is easily proven by Stirling's approximation of the factorial (see Lemma~\ref{cor:density}), and thus allows us to conclude that the incorporation of ties into the observable population space dramatically increases the set of observable, and thus uniquely stochastically measurable, random variables. In turn, the density (and thus linear continuity) of the empirical estimator solution space is also dramatically increased. 

In particular, this allows us to consider the space of discrete non-dichotomous, or ordinal by \textcite{stevens1946} random variables as well as continuous ones, subject to certain assumptions -- a long-standing excluded domain for analytically solvable non-parametric estimators. To assess distances between random variables expressible upon \(\mathcal{M}_{n}\), we introduce the \textcite{kemeny1959} distance, which satisfies the necessary properties of symmetry between pairs, sub-additivity, and the principle of indiscernability between all pairs \((\kappa_{X^{n \times 1}},\kappa_{Y^{n \times 1}}) \in \mathcal{M}_{n}\). However, the mathematical formulation introduced by Kemeny was similar to that of \textcite{kendall1938}, although not particularly efficient. In both instances, the distance functions are identified solely as Banach norm-spaces without an inner-product distance function, and are therefore defined as pre-Hilbert (or Banach) spaces, only when ties are, almost surely, not expected to occur. 

The following distance function \(d_{\kappa}\) is constructed as the Hadamard inner-product of two independently arising skew-symmetric matrices which present the linear permutation space upon a basis \(\kappa: \overline{\mathbb{R}}^{n} \to n \times n\), indexed \(k,l = 1,\ldots,n\) for each data matrix column, or random variable, \(j \in \mathbf{X}_{n \times p}, j = \binom{p}{2}\), for finite samples of size \(n \in \mathbb{N}^{+}\):

\begin{subequations}
\begin{equation}
\label{eq:kem_dist}
d_{\kappa}(X,Y) = \frac{n^{2}-n}{2} + \sum_{k,l=1}^{n} \kappa_{kl}(X) \odot\kappa_{kl}^{\intercal}(Y),~ k,l = 1,\ldots,n.
\end{equation}
\begin{minipage}{.45\linewidth}
\tiny
\begin{equation}
\noindent
\label{eq:kem_score}
\kappa_{kl}(X) = {
\begin{dcases}
\: \sqrt{.5} & \text{if } X_{k} > X_{k+1}\\
\: 0 & \text{if } X_{k} = X_{k+1}\\
\: -\sqrt{.5} & \text{if } X_{k} < X_{k+1}\\
\end{dcases}
}, 
\end{equation}
\end{minipage}
\begin{minipage}{.45\linewidth}
\begin{equation}
\scriptsize
\kappa_{kl}(Y) = {
\begin{dcases}
\: \sqrt{.5} & \text{if } Y_{k} > Y_{k+1}\\
\: 0 & \text{if } Y_{k} = Y_{k+1}\\
\: -\sqrt{.5} & \text{if } Y_{k} < Y_{k+1}\\
\end{dcases}
}, k = 1,\ldots,n-1.
\end{equation}
\end{minipage}
\end{subequations}

The \(\kappa\) function maps each extended real vector of length \(n\), sampled \(i.i.d.\) from a common population, onto the corresponding skew-symmetric matrix of order \(n \times n\), \(\kappa: \overline{\mathbb{R}}^{n \times 1} \to (n \times n)\). The matrix is skew-symmetric with only \(\tfrac{n^{2}-n}{2}\) free elements upon the sample, for the upper-triangle is a negation of the lower, and the diagonal is always 0, satisfying the necessary Gram-Schmidt principles of linear independence for a spanning basis. Each entry in said \(\kappa\) matrix corresponding to vectors \(\mathbf{A}_{kl} \in \kappa(X),B_{kl} \in \kappa(y)\) are denoted by an entry in the \(k^{th}\) row and \(l^{th}\) column, respectively. Moreover, \(X_{i} \equiv \kappa_{k,} = -\kappa_{,l},~ i,k,l = 1,\ldots,n\), as the rows and columns denote the relative ordering of each \(i^{th}\) observation upon the observed sample (and thus bijectively relate the rank and score upon the sample).  Upon these two \(\kappa\) skew-symmetric matrices is then performed the linear combination via the Hadamard, or element-wise, inner-product (\(\odot\)), each representing the bivariate variable pair and the linear combination of the corresponding element observation pair, whose inner-product is then summated over all \(n \times n\) elements. The distance calculated by the linear combination of the ordered vector space which results from the \(\kappa\) mapping, i.e., the permutation representation, is defined upon the a complete space of cardinality \(\mathcal{M} = n^{n}-n\), rather than \(n!\). The matrix transpose is denoted by the superscript \((\cdot)^{\intercal}\). By the transposition of the skew-symmetric matrix, the upper-triangle is merely the negation of the transpose of the defined lower-triangle, and thus also defined by equation~\ref{eq:kem_score} for all \(l = 1,\ldots,n\). Here, we assume that the probability of sampling is uniform upon \(\mathcal{M}_{n}\) and thus each pairing of permutations upon the population occurs uniquely. 

\begin{lemma}
\label{lem:unique_0}
The Kemeny distance function in equation~\ref{eq:kem_dist} between two identical random variables is always 0: \(d_{X,X}=0\) for \(\kappa(X) \in \mathcal{M}_{n}, n \in \mathbb{N}^{+}\).
\end{lemma}
\begin{proof}
By equation~\ref{eq:kem_dist} every element cell is combined with its negation: \(a_{kl} \cdot b_{lk} = a_{lk} \cdot b_{kl}\), noting the transposition which switches the rows and columns. If \(X = X\) then each \(k^{th}\) row in one matrix is the negation of the corresponding \(l^{th}\) column of the complementary matrix, and each element upon the Hadamard product is then the sum of the inner-product of the corresponding identical rows and columns. Under these established rules, each row-column vector Hadamard product in the skew-symmetric matrix must sum to 0, as all elements are identically ordered, and each row is the negation of the corresponding column, and is therefore always tied with itself. As each cell \(a_{kl} = -b_{lk} =0\), due to the assumption \(X=Y\), the sum of all \(n^{2}-n\) elements of the Hadamard product matrix is thus always \(0\cdot (n^{2}-n)\), as the lower-triangle is always the negative symmetric valuation of the upper-triangle, and therefore \(d_{\kappa}(X,X) = 0, \forall \kappa(X) \in \mathcal{M}_{n}.\) This also identifies the linear basis to be the \(\kappa\) mapping of the random vector variable, rather than the \(n \times 1\) vector as traditionally viewed.
\end{proof}

\begin{lemma}
\label{lem:positivity}
The Kemeny distance function in equation~\ref{eq:kem_dist} between any two points \(\rho_{X,Y}, \kappa(X) \ne \kappa(Y),\) is strictly positive.
\end{lemma}
\begin{proof}
Let there exist \(X^{n \times 1},Y^{n \times 1},Z^{n \times 1} \in \mathcal{M}_{n}, n \in \mathbb{N}^{+}\) and allow \(Z = X, Z \ne Y\) such that \(d_{\kappa}(X,Y)\le d_{\kappa}(X,X) + d_{\kappa}(X,Y).\) As \(X=Y\), \(0 \leq d_{\kappa}(X,Z) + d_{\kappa}(X,Z)=2d_{\kappa}(X,Z)\), and thus \(d_{\kappa}(X,Z) > 0\), subject to the assumption that \(X\ne Z\), whereupon Lemma~\ref{lem:unique_0} applies.
\end{proof}

\begin{lemma}
\label{lem:symmetry}
The Kemeny distance function in equation~\ref{eq:kem_dist} satisfies the symmetry property of a metric space.
\end{lemma}
\begin{proof}
Let \(\kappa(Z) = \kappa(X) \in \mathcal{M}_{n}\), such that \(d_{\kappa}(X,Y) \leq d_{\kappa}(X,X)+d_{\kappa}(Y,X) = d_{\kappa}(Y,X)\). Then, \(d_{\kappa}(Y,X) \leq d_{\kappa}(Y,Y) + d_{\kappa}(X,Y) = d_{\kappa}(X,Y)\), and \(d_{\kappa}(X,Y) = d_{\kappa}(Y,X)\).
\end{proof}

\begin{lemma}
\label{lem:sub-additivity}
The Kemeny distance function in equation~\ref{eq:kem_dist} satisfies the strong triangle inequality.
\end{lemma}
\begin{proof}
Allow \(\kappa(X),\kappa(Y),\kappa(Z) \in \mathcal{M}_{n}, n \in \mathbb{N}^{+}.\) By the uniqueness characteristic of a metric space, if \(\kappa(X)=\kappa(Y)=\kappa(Z)\), then \(d_{X,Z} \le d_{X,Y} + d_{Y,Z} = 0 \le 0 + 0,\). Next, allow \(\kappa(X) = \kappa(Y) \ne \kappa(Z),\) and let \(q = d_{X,Z}\). Then, \[d_{X,Z} \le d_{X,Y} + d_{Y,Z} = q \le 0 + q.\] Alternatively, if instead \(\kappa(X) = \kappa(Z) \ne \kappa(Y).\)  \[q = d_{X,Y}, d_{X,Z} \le d_{X,Y} + d_{X,Z} = 0 \le q + 0,\] which implies that \(0 < q\), as required. In either scenario, distance between all pairs is strictly less than the sum of two subsets upon the interval \([0,q]\), thereby ensuring sub-additivity upon \(\mathcal{M}_{n}, n \in \mathbb{N}^{+}.\)

Assume \(a,b\) then denote real distances upon equation~\ref{eq:kem_dist}. Let \(a = \pm p^{m}r\) and \(b = \pm p^{n}s,\) wherein \(r,s\) are rationals whose numerators and denominators do not share a common factor \(p\), and \(m,n\) are integers.As \(\|a\|_{p} = p^{-m},\|b\|_{p} = p^{-n},\) we further assume without restriction that \(m \le n\). This implies that \(p^{m}\le p^{n},\) such that \(p^{-m}\ge p^{-n}\)  must follow as well, and therefore \(\max\{\|a\|_{p},\|b\|_{p}\} = p^{-m}.\) For the general values of \(a,b \ne 0,\) we must show that \(\|a+b\|_{p}\le p^{-m}.\) Therefore, \(a+b=\pm p^{m} r\pm p^{n}s=\pm p^{m}\left(r\pm p^{n-m}s\right)\), whereupon \(n-m\ge 0\). Re-express \(r\) and \(s\) as fractions in lowest terms such that \(r=\frac{i}{j}\) and \(s=\frac{k}{\ell}\), where \(p\) does not divide \(i,j,k\), or \(\ell\). Then

\[r\pm p^{n-m}s=\frac{i}j\pm\frac{p^{n-m}k}\ell=\frac{i\ell\pm p^{n-m}kj}{j\ell}\;.\]

As there is no factor of p in this expression, then if \(n-m>0, \pm p^m(r\pm p^{n-m}s)$ expresses $a+b$ as a product of a power of $p$ and a rational whose numerator and denominator in lowest terms have no factors of $p$, and by definition \(\|a+b\|_p=p^{-m}\), then the strong triangle inequality is not violated. Should \(n-m=0\), the numerator in the last fraction of $(1)$ might have factors of $p$, but if so, then $a+b=\pm p^{m^{\prime}}t$ for some $m^{\prime}>m$, where $t$ in lowest terms has no factors of $p$ in either numerator or denominator, and $\|a+b\|_p=p^{-m^{\prime}}<p^{-m}$, which continues to satisfy the strong triangle inequality. Finally,  allow \(\|a + b\|_{p} = \|\kappa(a) \odot \kappa(b)\|.\) Then, without restriction, the logic of the proof is seen to hold, and therefore the Kemeny metric does satisfy the strong triangle inequality.  

\end{proof}
\begin{lemma}
\label{lem:hilbert}
The metric function in equation~\ref{eq:kem_dist} is a Hilbert space.
\end{lemma}
\begin{proof}
A Hilbert space is a metric space which possesses an inner-product formulation. As the inner-product of equation~\ref{eq:kem_dist} has been shown to possess non-negativity (Lemma~\ref{lem:unique_0} and Lemma~\ref{lem:positivity}), symmetry (Lemma~\ref{lem:symmetry}) and sub-additivity (Lemma~\ref{lem:sub-additivity}), it is a complete metric space, and therefore a Hilbert space.
\end{proof}

As a Hilbert space (Lemma~\ref{lem:hilbert}), a distance may be considered an affine-linear function upon a given domain. We note that the domain, \(\mathcal{M}_{n}\) consists of any vector of finite length which we assume to be independently and identically distributed upon the extended real line, \(X,Y \in \overline{\mathbb{R}}^{n \times 1}, n \in \mathbb{N}^{+} < \infty^{+}:\) there, ties upon the limiting values of \(\pm \infty\) are realised such that all finite elements are greater than \(\infty^{-}\), all negative infinite values are tied, all finite values are less than \(\infty^{+}\) and all positive infinite values are tied. This allows us to measure uniquely (by the Riesz representation theorem, valid for any Hilbert space) valuations of random variables which are considered degenerate upon the Minkowski norm \(\ell_{p}\)-space, such as Cauchy random variables. The co-image of equation~\ref{eq:kem_dist} is defined as the neighbourhood \(\{U_{m}(\mathcal{M})\}_{m=1}^{n^{n}-n} = [0,1,\ldots,n^{2}-n],\) consisting of the set of all permutations indexed by the \(m\) elements, each of which possesses a finite distance upon \(U(\mathcal{M}_{n})\). It should be noted that the affine-linear function space of the Kemeny distance is defined to extend to include monotonic invariance under transformation in addition to the standard \(a + bx\) definition of translation and scaling invariance: this is observed by the equivalence of \(\kappa(X) = \kappa(g(X))\) for any \(g(\cdot),\) comprised of the set of monotone transformations, and is due to the binomial construction of the skew-symmetric \(\kappa\) matrix in equation~\ref{eq:kem_score}. As long as all pairs of observations upon a random variable may be validly scored, or more technically, ranked in comparison, the underlying matrix itself remains unchanged. This allows us to conclude that the existence of a common population distribution \(F_{X},F_{Y}, (X^{n \times 1},Y^{n \times 1}) \in \overline{\mathbb{R}}^{n \times 2},\) is a sufficient condition for obtaining unique (by the Riesz representation theorem) regular probabilistic measurement structure (Lemma~\ref{lem:haar}) upon this linear Hilbert function space.

We next prove the Borel-Cantelli lemma upon the Kemeny metric space, establishing the almost sure finite convergence for all \(n \to \infty^{+},\) thereby including the observation of a sample which is bijectively equivalent to the population. As a Hilbert space upon a neighbourhood of distances \(U(\mathcal{M}_{n})\), which is of diameter \(n^{2}-n\), the expectation \(E(U_{\mathcal{M}_{n}})\) must be almost surely finite upon any finite sample or population, as long as the elements are sampled independently.

\begin{lemma}
\label{lem:borel-cantelli}
The Borel-Cantelli exists and ensures almost sure convergence for all \(n \in \mathbb{N}^{+}\) upon the Kemeny metric.
\end{lemma}
\begin{proof}
A sufficient conditions exists in that the Kemeny metric is a Hilbert space (Lemma~\ref{lem:hilbert}) and thus almost surely converges, as the diameter support is finite with probability 1 for all finite \(n\). For all finite \(n\), all distances are no greater than \(n^{2}-2\), and thus all distances are by definition finite upon \(\mathcal{M}_{n}\). Convergence in the limiting condition is now examined, for \(n \to \infty^{+},\) wherein the support upon \(U(\mathcal{M}_{n})\) is defined as \([-\tfrac{n^{2}-n}{2},\tfrac{n^{2}-n}{2}]\). The extremal permutations are observed twice upon this neighbourhood, reflecting the diameter of the graph as of length \(n^{2}-n\), which is then degenerate: all other permutation pairs are thus of finite distance, as by sub-additivity, all other points are on the interior. Therefore, if for any \(n\) there exists at most 2 conditions upon \(n^{n}-n\) events which may possess the maximal distance, finite convergence is ensured to hold almost surely for all \(n\) as:
\begin{align*}
\Pr(E(|d_{\kappa}(X,Y)|)<\infty^{+}) & = \lim_{n\to\infty^{+}}  \frac{n^{n}-n - 2}{n^{n}-n}  + \lim_{n\to\infty^{+}} \frac{2}{n^{n}-n}\\
                                     & = 1.
\end{align*}  
\end{proof}

\subsubsection{Moment properties of the Kemeny distance}

The Kemeny distance constructs a linear function space upon bivariate random variables constructed from \(n\) observations, resulting in an almost surely finite scalar distance. Optimisation upon this domain of permutations, and their representation as a graph (a necessary consequence of the ultrametric properties of Lemma~\ref{lem:sub-additivity}) relies upon exhaustive searching, which for our problem would almost surely result in a combinatorial explosion for relatively small sample sizes. The bivariate estimator problem for a single pairwise distance is feasible, and with the incorporation of sampling assumptions, the estimator properties of this distance may be defined. The Kemeny space is itself an affine-linear function space though, and therefore conveys a number of desirable properties. This work presumes, a priori, that all \(X^{n \times 1},Y^{n \times 1} \in \overline{\mathbb{R}}\), for \(n \in \mathbb{N}^{+},\) such that is included random variables upon degenerate distributions such as the Cauchy distribution, all of which arise under independent sampling from a common population, unless explicitly stated otherwise. This allows us to consider a unique solution upon conventionally avoided degenerate yet uniformly independently sampled distributions, such as the Cauchy distributed random variables.

Upon any finite sample \(n\) then is observed by \(i.i.d.\) a bivariate random variable whose distance is a finite valued realisation by equation~\ref{eq:kem_dist}. Any finite sample is presumed to be empirically realised with stochastic error, and therefore represents an imperfect estimate of the unique underlying population parameter, \(\lim_{m \to \infty^{+}} \hat{\theta} \approx \theta,\) where \((\hat{\theta},\theta)\) are upon a Hilbert space and therefore affine-linear transformations of the Hilbertian distance function. We must first show that equation~\ref{eq:kem_dist} is an unbiased estimator, and thus all affine-linear transformations thereof are also unbiased upon the population neighbourhood of distances, \(U(\mathcal{M}_{n}).\)  A similar conclusion may be found in \textcite[Part III]{kemeny1959}, wherein the probability bound is equivalent to our bound as the median value of the CDF. Let \(m \to \infty^{+}\) be understood to reflect, for any \(n\), the limiting population of permutations \(n^{n}-n\).


\begin{lemma}
\label{lem:unbiased}
The Kemeny distance and affine-linear transformations thereupon (e.g., equation~\ref{eq:kem_cor}) are unbiased estimators, for all \(X,Y \equiv (\kappa(X),\kappa(Y) \in \mathcal{M}_{n}), n \in \mathbb{N}^{+}.\)
\end{lemma}
\begin{proof}
Assume the existence of a test statistic \(T_{n},\) with \(E_{H_{0}}(T_{n})=0\), and a finite positive variance normalised to 1 without loss of generality \(\sigma^{2}_{T_{n}\mid H_{0}}=1,\) whose existence is guaranteed by Lemma~\ref{lem:borel-cantelli}. The null hypothesis \(H_{0}\) merely denotes that the an estimator class possesses an expectation of 0 and unit variance, expressed as \[d_{\kappa}(X,Y) = \frac{T_{n} + \mu_{1}}{\sigma^{2}_{T_{n}\mid H_{0}}}.\]

The population of the Kemeny distance upon any finite \(\mathcal{M}_{n}\) is of distance \(0 \le \tfrac{n^{2}-n}{2} \le n^{2}-n, n > 1.\) As the extremum occur uniquely only once upon any \(U(\mathcal{M}_{n})\) it follows by the pigeon-hole principle and sub-additivity that the set \(\mathcal{M}_{n}-2\) must always be on the interior of the valid support. For \(n=2\), the set of \(\mathcal{M} = 2^{2} - 2 = 2\) permutations which possess a symmetric (by the even function nature of equation~\ref{eq:kem_dist}) neighbourhood of distances in \(n^{2} -n + 2\cdot\{-(\sqrt{0.5})^{2},+(\sqrt{0.5})^{2}\}\), thereby producing distances of \(\{0,2\} = \{0,n^{2}-n\}\), whose expectation is thus \(E(d_{\kappa}(X,Y)) = \tfrac{n^{2}-n}{2}=1\). When the set of all elements \(m\) is summated over, the population expectation \(\lim_{m \to\infty^{+}} E_{U_{m}} = \tfrac{n^{2}-n}{2},\) which holds inductively for any \(\mathcal{M}_{n}\).

By induction wrt \(m\), observe the finite telescoping sequence of \(m_{i} \in \mathcal{M}, i = 1,\ldots,n^{n}-n\), from which follows the finite expectation of 0 by the closure under addition for the skew-symmetric matrix of equation~\ref{eq:kem_score} in the Kemeny metric space. This is true by the symmetry of the telescoping positive and negative distances, guaranteed by the even function nature of the \(\kappa\) function centred at the arbitrary point of origin (see Lemma~\ref{lem:even}), which occur with equal finite frequency. For any finite \(n\), the sum and inner product of two random variables \(X,Y\) of any known finite \(n\), on the population \(\mathcal{M}_{n}\) indexed by \(m\mid{n}\), holds 
\begin{align*}
\mu_{1} & = \lim_{m\to \mathcal{M}} \tfrac{n^{2}-n}{2} + E_{m}(\rho_{X,Y}) \\
        & = \tfrac{n^{2}-n}{2} + \sum_{k,l = 1}^{n} \kappa_{kl}(X_{m})\odot\kappa_{kl}^{\intercal}(Y_{m}) \\
        & = \tfrac{n^{2}-n}{2} + 0,
\end{align*} via the Tchebvyshev inequality (Lemma~\ref{lem:chebyshev}).

Therefore, \(E(U(\mathcal{M}_{n})) = \tfrac{n^{2}-n}{2} = \mu_{1}\)
\begin{align*}
d_{\kappa}(X,Y) & = \frac{T_{n} + \mu_{1}}{\sigma^{2}_{T_{n}\mid H_{0}}}\\
\lim_{m \to\infty^{+}} E(d_{X,Y})\sigma^{2}_{T_{n}\mid H_{0}} - \mu_{1} & = E_{H_{0}}(T_{n}) = 0,\\
\end{align*}
and therefore the expectation of the distance function is an unbiased estimator upon the permutation population for any \(n\), as are any affine-linear or monotonic transformations thereupon.
\end{proof}
Transformation of the distance from the complete metric properties of the Hilbert space to a signed distance space is obtained by subtraction of the leading median distance, resulting in a finite subset \(m\) which in aggregate averages out to 0 in the neighbourhood of \(U(\mathcal{M}_{n})\) which is uniformly operated upon by the Glivenko-Cantelli theorem (Lemma~\ref{lem:gc}) for this linear function space. Further examination will demonstrate that the errors of approximation are themselves symmetrically distributed around the expectation (as the Kemeny metric is an even function; Lemma~\ref{lem:even}), and that by sub-additivity (Lemma~\ref{lem:hilbert}), the affine linear transformation by subtraction of the unbiased expectation does not bias the estimator, for all finite bivariate distances between random variables of length \(n\), thereby resolving problem 1.
\subsubsection{Variance of the Kemeny distance}
The Kemeny linear variance or concentration measure results from the summation of the \(n^{2}-n\) free parameters representing the finite strictly non-negative support for all \(\kappa(X_{n \times 1})\), and is necessary to solve problem 2 in the construction of a Wald test statistic. Consider
\begin{equation}
\label{eq:kem_variance}
\sigma_{\kappa}^{2}(X) = \frac{2}{n(n-1)}\Big(\sum_{k=1}^{n}\sum_{l=1}^{n} \kappa_{kl}(X)\kappa_{kl}(X) \Big) = \frac{2}{n(n-1)}\sum_{l,k=1}^{n}\kappa_{kl}^{2}(X).
\end{equation}
From this expression of the variance is defined a mapping of the extended real domain sub-space to a singular real which is positive if and only if \(X\) is non-degenerate, by simple algebraic properties. An easily established condition which would serve to simplify the characterisation of the Kemeny distances' variance over \(U(\mathcal{M}_{n})\) would be to show that the Kemeny distance function satisfies the central limit theorem, which we now do:

\begin{lemma}~\label{lem:clt_kem}
Let \(\{x_{1},\ldots,x_{n}\}\) be a random sample of size $n$ drawn from a distribution of expected value given by $\mu$ and finite variance given by $\sigma^{2}_{\kappa}$. By the law of large numbers upon the Kemeny metric function space (Lemma~\ref{lem:lsn}), the sample converges in probability to the expected value $\mu$ for the limit wrt $n$, as required under the Lindeberg-L\'{e}vy Central Limit Theorem.
\end{lemma}
\begin{proof}
The necessary conditions to validly apply the (generalised) central limit theorem to the strictly sub-Gaussian Kemeny metric are threefold. First, substitute $\lim_{n\to\infty^{+}} \frac{1}{n} \equiv |\mathcal{M}|^{-1} = (n^{n}-n)^{-1} = 0,$ a necessary conversion to obtain the atomless probability measure for any finite \(m \mid n\), which is easily observed to hold for any sufficiently large \(n\) by the partial derivative upon the growth function of \(n\). This produces a complete measure space of probability \(\tfrac{1}{|\mathcal{M}|}\) (see Lemma~\ref{lem:haar}), for which the probability of non-observability tends to 0, and the probability of observation of a point which is non-measurable over \(\mathcal{M}_{n}\), by the Haussdorff compactness of the population wrt \(m\) also tends to 0 almost surely. Thus the entire non-degenerate extended real domain is linearly measurable: \[1 -\lim_{m\to\infty^{+}} \Pr(\mathcal{M} \setminus \kappa(x)) = 1.\]

This is confirmed by the Markov inequality, which in the absence of any non-measurable finite vectors upon the mapping function \(\kappa\) also which tends to 0 for any sequence of increasing length \(n\). Let $\mu^{1}_{\kappa}$ exist (Lemma~\ref{lem:borel-cantelli}) for the Kemeny metric, establishing an almost sure first expectation upon any finite population, which is strongly observed with expected error of 0 (Lemma~\ref{lem:unbiased}). Finally, let the variance be finite for any population, which has also proven valid upon the Kemeny metric space (Lemmas~\ref{lem:kem_bounded} and \ref{lem:chebyshev}). By these three conditions it therefore follows that the Central Limit Theorem is a valid linear construction upon the Kemeny metric space, thereby ensuring linear convergence with strong (sub-Gaussian) probability for any sequence of $p$ unknown yet identified \(i.i.d.\) estimated abelian functions, which are strictly sub-Gaussian distributed variables for any \(n\).
\end{proof}

As the support of the population of all permutations \(\mathcal{M}_{n}\) is recursive yet distinct for each \(n\), under the assumption of independent and identical sampling results an expression for the population variance function for all distances upon equation~\ref{eq:kem_dist} upon \(\mathcal{M}_{n}\): 
\begin{equation}
\label{eq:kem_population_variance}
\dot{\sigma}_{\kappa}^{2}(\mathcal{M}_{n}) = \frac{(n - 1)^2 (n + 4) (2 n - 1)}{18 n},
\end{equation}
representing the population variance of all Kemeny distances of a fixed length \(n\) between bivariate independent random variables. Note that this domain of \(U(\mathcal{M}_{n})\) explicitly presumes that permutations are observed only once, and we will relax this assumption in a following paper, in order to allow for Studentification upon the test statistics' distribution for non-parametric estimators. 

For sufficiently large populations, numerically confirmed to be \(\mathcal{M}_{n\ge 9}\), the observed properties comply with the theoretical guarantees of the Central Limit Theorem (Lemma~\ref{lem:clt_kem}). This restriction is due to the failure of the central limit theorem upon very small samples, wherein the expected distribution is symmetrically bimodal rather than uni-modal at distance 0 thereby contradicting the definition of the regular uni-modal structure upon the neighbourhood of all distances. The bias of expression~\ref{eq:kem_population_variance} is limited and curtails sharply, beginning with samples of 9 or more, and approximates larger samples' estimates of the population variances quite accurately, and thus is adequate for general purposes.

We now show that the Kemeny distance estimator also satisfies the Gauss-Markov theorem, providing a best linear unbiased estimator which satisfies the Lehmann-Scheff\'{e} theorem (Definition~\ref{def:gauss_markov}):

\begin{definition}[Lehmann-Scheff\'{e} theorem]
\label{def:gauss_markov}
The necessary conditions to satisfy Gauss Markov theorem guarantee that the distance minimising function provides the best linear unbiased estimate (BLUE) possible point estimates upon a given sample. The five necessary Gauss Markov conditions are:
\begin{enumerate}
\footnotesize{
    \item{Linearity: estimated parameters must be linear.}
    \item{Variables arise i.i.d. by stochastic sampling from a common population.}
    \item{No variables are perfectly correlated.}
    \item{Exogeneity: the random variables are conditionally orthonormal.}
    \item{Homoscedasticity: the error of the variance is constant (homogeneous) across the given population.}
}
\end{enumerate}
\end{definition}

\begin{theorem}
\label{thm:gauss-markov}
The distance function in equation~\ref{eq:kem_dist} is a Gauss-Markov estimator for any bivariate vector pair of length \(n\) which are independently and identically sampled upon a common population.
\end{theorem}
\begin{proof}
The required linearity (1) and lack of collinearity (3) of the function space for domain \(\mathcal{M}_{n}\) follows by definition for any Hilbert norm space, and are satisfied for the equation~\ref{eq:kem_dist} and functions thereof by Lemma~\ref{lem:hilbert}; homogeneity upon domain \(\mathcal{M}_{n}\) are axiomatically given. Condition 2 is obtained via axiomatic assertion for the proof. Unbiasedness or exogeneity follows from Lemma~\ref{lem:unbiased}, and Gramian positive definiteness follows from either the satisfaction of the Mercer condition, or as the finite sum of the squared totally bounded positive variances of the Kemeny metric space (Lemma~\ref{lem:kem_bounded} enacted upon \(\mathcal{M}_{n}\); \cite{schoenberg1938}); both conditions are equivalent, and thus valid for the Kemeny measure space. The exclusion of the \(n\) degenerate permutations completes the space to ensure a unique conditional linear independence. Utilisation of Bochner's theorem to guarantee the Gramian nature is also valid, recognising that the Hilbert space is a continuous positive-definite function on a locally compact abelian group. 

Equivalently, the correlation between the two vectors must solely depend upon the distance between them, while remaining a.s. positive definite (i.e., a valid stochastic distance function) for any domain as defined by assumption. This property is valid for the Kemeny distance as an unbiased linear estimator and the removal of the two collinear permutation points upon \((U(\mathcal{M}_{n}))^{2} >0\) guarantees a population of strictly positive definite finite distances are always observed upon \(\mathcal{M}_{n}\). Therefore, the correlation (and therefore covariance) matrix is always positive definite, for any bivariate random population of permutations which are neither collinear or degenerate.
Condition 3 in Definition~\ref{def:gauss_markov} is satisfied by axiomatic assumption, wherein the removal of the unique points \(\mathcal{M}_{n} \setminus \{\kappa(I_{n}),\kappa(I^{\prime}_{n})\}\) defines set \(Q\), with non-collinear uniqueness guaranteed by the Riesz representation theorem for a Hilbert space. For arbitrary random sampling upon the Kemeny neighbourhood, there exist only 2 \(\kappa\) mappings \(\mathcal{M}_{n} \setminus Q = 2\)  which are collinear for any pair of independent random variables. The probability of repeated observation (see also Lemma~\ref{lem:clt_kem}) quickly then tends to 0 in the limit wrt \(n\): \[ \Pr(\mathcal{M}_{n} \cap {Q}) = \lim_{n\to\infty^{+}} \frac{|\mathcal{M}_{n} \setminus Q| = 2}{n^{n} - n} = 0,\] and in general we assume by sub-additivity that all perfectly linear permutation basis may be excluded uniformly, as the probability of the complementary event occurs with probability 1. Exogeneity in turn follows from the Riesz representation theorem for any Hilbert space, as independent random sampling follows under axiomatic assumption of the theorem domain such that \(\lim_{n \to \infty^{+}} \tfrac{|\mathcal{M}_{n} \setminus Q| = 2}{n^{n}-n} \to 0\), and as the affine-linear function space is a sufficient unique representation, the distance is a sufficient optimal estimator of the true Kemeny distance.

The homoscedasticity assertion follows by definition of the Kemeny variance \(0<\sigma^{2}_{\kappa}<\infty^{+}\) whereupon only linearly comparable scores may be ordered, as there may only exist one set of permutations upon any common population function (Cumulative Distribution Function; CDF) by axiomatic assumption. This explicitly removes the valid operationalism upon periodic functions, which is valid by axiomatic assumption, else there must exist realisations upon the extended real line which may not be validly assessed by equation~\ref{eq:kem_score}. This thus completes the proof that the Kemeny correlation satisfies all necessary requirements of a Gauss-Markov estimator, which exhibits convergence in probability under the strong law of large numbers (Lemma~\ref{lem:lsn}).

\end{proof}
This is an important, because it allows for a minimum variance local linear estimator to be defined upon random variable sampled \(i.i.d.\) from the extended real line, without loss of generality, thereby providing a linear solution of the median expected score value, rather than the mean. It is a necessary condition of course that a Gauss-Markov estimator be capable of obtaining the Cram\`{e}r-Rao lower bound, which we will now proceed to prove:
\begin{lemma}
\label{lem:cramer-rao}
The Kemeny estimator functions satisfy the Cram\`{e}r-Rao lower bound upon the population \(\mathcal{M}_{n}\) constructed of asymptotic limit on \(n\).
\end{lemma}
\begin{proof}
The unbiasedness of the estimator function is established in Lemma~\ref{lem:unbiased}, and said estimator function observed to be asymptotically normally distributed as well by Lemma~\ref{lem:kem_asym_normal}. As the variance of the estimator function is strictly sub-Gaussian for all finite \(n\), the variance is a scalar constant ratio which converges to 1 from below as a linear function of all data distributions. Under these conditions, it follow that by the Gauss-Markov theorem~\ref{thm:gauss-markov} the asymptotic variance approaches from below to the asymptotic variance of the normal distribution. In the limit wrt \(n\), this variance is 0 in expectation, and thus equivalent, and is otherwise strictly smaller by the existence of a compact and totally bounded domain subset of the real line for any variance upon \(U(\mathcal{M}_{n})\). This therefore concludes the proof in obtaining the Cram\`{e}r-Rao lower bound \[\lim_{n\to\infty^{+}} \frac{n^{n}-n}{n^{n}}\sigma^{2}_{\kappa} \to \sigma^{2},\]
as is necessary by the central limit theorem (Corollary~\ref{lem:clt_kem}). Asymptotic convergence is guaranteed wrt \(n\) by Lemma~\ref{lem:lower_equality}, which holds that for any random variable upon the Kemeny metric, the error distance will converge to 0 in expectation, as an isometric solution under the Euclidean distance, thus establishing collinearly upon the asymptotic wrt \(n\) by the Riesz representation theorem. The upper-bound upon the measurability ensures that even if the error convergence does not tend to 0 (as would occur when a model is incorrectly specified, and thus the conditional Bayes error rate upon the sample is greater than 0), the error is still always uniquely identified, and the local sample estimator converges to the population expectation quadratically.
\end{proof}

These properties demonstrate the duality construction of the two estimator function spaces of the Kemeny and Frobenius norms to be complementary, rather than mutually exclusive, as a satisfaction of locally optimal unbiased and minimum variance approximation of the population bilinearity condition wrt both rank and score.
%
%
%
\subsection{Kemeny correlation estimator}

Transforming the distance function (equation~\ref{eq:kem_dist}) via an affine-linear manipulation into an inner-product correlation coefficient is achieved in equation~\ref{eq:kem_cor},
\begin{equation}
\label{eq:kem_cor}
\tau_{\kappa}(X,Y) = -\tfrac{2}{n^{2}-n}\sum_{k=1}^{n}\sum_{l=1}^{n} \kappa(X)_{kl}\odot\kappa^{\intercal}(Y)_{kl} 
\end{equation}

By taking the finite, compact, and totally bounded neighbourhood about 0, the distance \(d_{\kappa} \in [0,n^{2}-n]\), for arbitrary \(n\), and its support of the extended real domain of two extended real number lines of length \(n\), is transformed to a closed and compact neighbourhood centred at 0, \(U(\mathcal{M}_{n}) = [-\tfrac{n^{2}-n}{2},\dots,0,\dots,\tfrac{n^{2}-n}{2}]\), s.t. \(E_{U(\mathcal{M}_{n})} = 0\). Multiplying by 2 rescales the extrema to span \(\pm (n^{2}-n)\) without changing the expectation of 0, and the negation ensures that the larger the a.s. finite realised distance from the arbitrary affine-linear origin of 0, the closer to a reverse permutation the random variable is. The extremal distance \(\tfrac{2(n^{2}-n)}{2(n^{2}-n)} = 1 \) upon the diameter under negation becomes \(-1\). Thus, equation~\ref{eq:kem_cor} is consistent with the definition of a Hilbert metric space upon the set of all non-degenerate permutations with ties, and is a Gauss-Markov estimator under independent sampling from a common population.


\subsubsection{Probability regularity of the Kemeny distance and all affine-linear transformations}
Assume without loss of generality that the extended real vector space of \(\mathbf{X}^{n \times p}\) variates are expressible upon a positive definite variance-covariance matrix \(\Xi_{p\times p}(\mathbf{X})\) constructed from equations~\ref{eq:kem_cor} and \ref{eq:kem_variance}. With the fulfilment of the definition of the observation of a positive finite variance for any non-constant random sequence, i.e., a variable, along with the strictly positive support of the neighbourhood about 0 for the squared signed distances on \(U(\mathcal{M}_{n})\) (equation~\ref{eq:kem_variance}), the distribution of the distances upon \(U(\mathcal{M}_{n})\) satisfies the necessary conditions to be considered sub-Gaussian \parencite[Ch.~1]{buldygin2000} by the regular probability measure upon \(U(\mathcal{M}_{n})\) as guaranteed for any Hilbert space. This is realised by the sufficient and necessary conditions of the compactness and totally bounded Haussdorff space upon the uniformly iid sampled Kemeny metric.

The expansion of the minimum necessary moment sequence (i.e., \(n\) is not a sufficient population parameter) is now discussed. A sub-Gaussian variable \(\xi^{n \times 1}\) is said to be strictly sub-Gaussian if and only if
\begin{definition}
\label{def:stict_sg}
\begin{equation}
\label{eq:strict}
E(e^{\lambda\xi}) \le e^{\frac{\lambda^{2}\sigma^{2}_{\xi}}{2}}, \forall~ \lambda \in \overline{\mathbb{R}} \equiv \|\xi\|_{\kappa} \le \sigma_{\xi} = \|\xi\|\ell_{2}(\Omega).
\end{equation}
\end{definition}

\begin{lemma}
\label{lem:strict_subgauss}
The population distribution of the Kemeny distances upon \(U(\mathcal{M}_{n})\) for any finite sample \(n < \infty^{+}\) are strictly sub-Gaussian (Definition~\ref{def:stict_sg}) and is therefore stable.
\end{lemma}

\begin{proof}
The distribution of the Kemeny distance is, for all \(n\), defined upon the neighbourhood \(U(\mathcal{M}_{n})\), and is therefore almost surely finite and measurable (the existence of a perfectly normal, henceforth \(T_{6}\), Hilbert space ensures the existence of a suitable Borel \(\sigma\)-measure). The finite definiteness of all possibly observable variances has also already been established by either Lemma~\ref{lem:kem_bounded} upon \(\mathcal{M}_{n}\), or by the Borel-Cantelli lemma~\ref{lem:borel-cantelli}. A sub-Gaussian variable is defined as the conjunction of a finite variance and compact totally bounded support, and therefore the Kemeny distance is sub-Gaussian for finite \(n\) \parencite[Ch.~1]{buldygin2000}. For either the bivariate or univariate case then, we observe a sub-Gaussian distribution: in the univariate scenario, the spectrum of the distance measure is \((0,b], b \in (0,\ldots,\mathbb{N}^{+},\infty^{+})\), with the trivial restriction that \(n>0\) to allow \(b \ne 0\), and for which the expectation under affine transformation of 0 satisfies Condition~\ref{eq:strict}. Note that a Gaussian random variable possesses a support wider than that of our function space, \([\infty^{-},\infty^{+}] \supset [-\tfrac{n^{2}-n}{2},\tfrac{n^{2}-n}{2}],\) for all finite \(n\).

For the bivariate scenario, the spectrum of the vector inner-product for the \(\kappa\) skew-symmetric matrices \(\{\kappa_{X},\kappa_{Y}\}\) is also finite and totally bounded, and is also therefore strictly sub-Gaussian for any finite \(n\), as their distance is almost surely less than or equal to the Euclidean linear distance \(\ell_{2}\)-norm uniformly across \(U(\mathcal{M}_{n})\). Thus, both the marginal distribution and bivariate distributions are strictly sub-Gaussian for all finite \(n\). For the contrapositive condition, consider if, for finite \(n\), the distribution of the population of Kemeny distances were Gaussian. If this were so, then there would exist measurable events which occur with non-zero density upon the Gaussian pdf but which are not observable for any finite \(n\):
\begin{multline}
\lim_{n\to\infty^{+}} \Pr(x \setminus U(\mathcal{M}_{n}) \in \pm \tfrac{n^{2}-n}{2}) =  \int_{\infty^{-}}^{-\tfrac{n^{2}-n}{2}} \frac {1}{\sigma {\sqrt {2\pi }}} e^{-{\frac {1}{2}}\left({\frac {x-\mu }{\sigma }}\right)^{2}}\diff{x} - \\ 
\int_{-\tfrac{n^{2}-n}{2}}^{\infty^{+}} \frac {1}{\sigma {\sqrt {2\pi }}} e^{-{\frac {1}{2}}\left({\frac {x-\mu }{\sigma }}\right)^{2}}\diff{x} \ge 0, ~x \in \mathbb{R}.
\end{multline}
As such, there must exist for all populations of permutations' observable distances upon the Gaussian probability distribution which do not occur upon the Kemeny distance support, and such equality is only obtained in the limit, whereupon equality to 0 holds, resulting in a paradox, as non-existent measures occurring with positive probability. The distribution of the population is consequently strictly non-Gaussian over the domain \(X^{n \times 1},Y^{n \times 1} \in \overline{\mathbb{R}}, \forall n < \infty^{+}\), and must converge to a Gaussian distribution only asymptotically in the limit wrt \(n\to\infty^{+}\).
\end{proof}

For this linear function space then, we establish the asymptotic normality of the distance and correlation estimator:
\begin{lemma}~\label{lem:kem_asym_normal}
The bivariate strictly sub-Gaussian Kemeny metric space is asymptotically normally distributed.
\end{lemma}
\begin{proof}
Proof by contradiction assures the asymptotic normality of the set \(U(\mathcal{M})\). Upon the population limit \(n\to\infty^{+}\), the cardinality of the set diverges from the real line and there is no collection of finite distances or cardinalities, as \(|U(\mathcal{M})| = n^{n}-n ~\forall n\) via Lemma~\ref{lem:borel-cantelli}. Mapping said ratio of the elements onto the finite support of the neighbourhood satisfies the upper-limit of Kolmogorov's stronger order property, yet also produces a degenerate bound upon an infinite set. Upon said set, the distances upon an infinite set of elements may no longer be uniquely ordered upon a finite distance relative to the unique origin \(I_{n}\), which is realised as equivalent to a probability measure of size \(0 = (U^{*}(\mathcal{M}))^{-1}\), and therefore cannot be stably strictly sub-Gaussian. By the strict bijection between rank and score measure distances, the Gaussian distribution is sufficient and uniquely capable of denoting said linear function space, as two identical measures are unnecessary (Lemma~\ref{lem:clt_kem}). Upon the population, by the weak law of large numbers exists a perfectly Gaussian random variable upon which the linear bijection of rank and score is observed to hold almost surely and bilinearly. Therefore, the divergent Kemeny distance is bijectively collinear with the Gaussian probability density function and its integral, once centred and scaled, and thus both measure spaces upon a common distribution are uniquely known by the Euclidean distance. This confirms the isometric equality of upon the populations of all Hilbert spaces to the Euclidean distance.
\end{proof}

\begin{lemma}
\label{lem:lsn}
The Kemeny metric satisfies the strong law of large numbers for any identically and independently distributed pair of independent random variables as a linear distance function.
\end{lemma}
\begin{proof}
By Markov's inequality and the Borel-Cantelli lemma, the existence of a finite expectation satisfies the strong law of large numbers, which is guaranteed for any finite \(n\) upon the Kemeny metric. This condition is equivalent to establishing that \(\sup{\sqrt{(\pm \frac{n^{2}-n}{2})^{2}}} = \frac{n^{2}-n}{2} \equiv \sup{d_{\kappa}(x,y)}, \,\forall~  n \in \mathbb{N}^{+}\), and is trivially verified. A second condition trivially holds that for any finite vector sequence \(\{x_{i}\}_{i=1}^{n}\), a linear ordering may be obtained, such as by using the image of the \(\kappa\) mapping from any random variable independently and identically sampled upon the extended real line. Thus, the strong law of large numbers is observed to hold for any finite sample upon \(U(\mathcal{M}_{n})\) for which \(\kappa(\cdot)\) may be validly measured.
\end{proof}

Finally, we now consider the nature of the Beta-Binomial measure space upon the bivariate random sample, which we will show to be a valid Haar measure.

\subsection{Haar measure}
\label{subsec:haar}
The existence of Haar measures allows us to define admissible procedures such that optimal invariant decision criteria may be established. More pragmatically, they allow us to take the probability structure guaranteed by the existence of a Banach norm-space, such as invariant transformation of optimisation functionals which are admissible for both Frequentist and Bayesian inference. In particular, Haar measures allow the construction of prior probabilities and conditional inference in statistics (thereby enabling extensions to function spaces including multiple regression). Thus, we establish here a number of mathematical primitives upon the Kemeny metric space, including the existence and validity of Radon derivatives and integrals over the invariant volume of the sample, rather than the population. Both properties are foundational within statistical analysis, and are now shown to be validly enacted upon the Kemeny function space as well. Further, we note that by the properties of the Haar measure, as \(\mathcal{M}_{n}\) and functions thereof are discrete, the countably additive left-invariant measure may be wholly defined on all subsets of the domain, by the axiom of choice.

A function as a Haar measure is defined as a unique countably additive, non-trivial measure \(\mu\) on the Borel subsets of \(G\) satisfying the following properties:
\begin{definition}~\label{def:haar}
\begin{multicols}{2}
\begin{enumerate}
    \item{The measure \(\mu\) is left-translation-invariant: \(\mu (gS)=\mu (S)\) for every \(g\in G\) and all Borel sets \(S\subseteq G\).}
    \item{The measure \(\mu\) is finite on every compact set: \(\mu (K)<\infty^{+}\) for all compact \(K\subseteq G\)}
    \item{The measure \(\mu\) is outer regular on Borel sets \(S\subseteq G\): \[\mu (S)=\inf\{\mu (U):S\subseteq U\}.\]}
    \item{The measure \(\mu\) is inner regular on open sets for compact \(K\) \(U\subseteq G:\)\[\mu (U)=\sup\{\mu (K):K\subseteq U\}.\]}
\end{enumerate}
\end{multicols}
A measure on \(G\) which satisfies these conditions is called a left Haar measure, and is a sufficient and necessary condition to establish right Haar measure existence and proportionality, and therefore equivalence.
\end{definition}

\begin{lemma}~\label{lem:radon}
The Kemeny metric space satisfies Definition~\ref{def:radon} of a Radon measure space, thereby proving for the existence of a Radon derivative (Def~\ref{def:radon}).
\begin{definition}~\label{def:radon}
If \(X\) is a Hausdorff topological space, then a Radon measure exists on \(X\) which is uniformly locally finite inner regular Borel measure \(m\) on \(X\).
\end{definition}
\end{lemma}
\begin{proof}
The Kemeny metric is a Hilbert metric space (by equation~\ref{eq:kem_dist}), and thus is a \(T_{6}\), or perfectly normal space Hausdorff topological vector space \(U(\mathcal{M}_{n})\), and therefore a Radon derivative on the continuous space exists.  Consequently, \(\kappa(X), X \in \overline{\mathbb{R}}^{n \times 1}\) must be both locally finite and inner regular, while also being a continuous measure . By the Riesz representation theorem, all metric spaces are inner regular on open sets \(K\), where \(K\) denotes the finite support of the Kemeny metric over all \(n\), which by Lemma~\ref{lem:kem_bounded} is also always locally finite. Thus the finite Borel measure upon the Kemeny metric is tight (in the sense of \cite{bogachev2007} Theorem 7.1.7), and there exists a Radon derivative upon the Kemeny measure.
\end{proof}

\begin{lemma}
\label{lem:haar}
The Kemeny metric space satisfies all properties of Definition~\ref{def:haar} and is therefore a Haar measure space over all random independent variables sampled independently and identically from the extended real line.
\end{lemma}
\begin{proof}
The Kemeny metric and its Borel \(\sigma\)-algebra are closed under addition and multiplication, and therefore are left-translation-invariant. The total boundedness of Lemma~\ref{lem:kem_bounded} guarantees the measure \(\kappa(K)\) is finite for all \(K \subset \overline{\mathbb{R}}\). The inner-regularity is proven in Lemma~\ref{lem:radon} and both outer regularity and completeness follow by axiomatic assertion upon a Hilbert space (Lemma~\ref{lem:hilbert}). Therefore by Definition~\ref{def:haar}, the Kemeny metric is a Haar measure.
\end{proof}

As a compound distribution, the Beta-binomially distributed distance function is difficult to computationally work with. However, it is noted that the binomial is defined upon known support for any \(U(\mathcal{M}_{n})\), and therefore may be removed without loss of generality, as it serves only to discretise the probability space into unit intervals upon a sample. Thus, the probability generating function of the non-parametric distance function is itself is Beta distributed, and much simpler to work with. Should we accept that a Beta-Binomial distribution is necessary to restrict the support to the interval of the \(\kappa\) function and thus the Kemeny metric function (Lemma~\ref{lem:kem_bounded}), then the univariate distribution of the variance must also be a central \(\chi^{2}_{\nu=1}\) distribution for \(\nu\) degrees of freedom, using the modified Bessel function of the first kind \(I_{v}\) of the form
\begin{equation}
\scriptsize{
f(X\mid \sigma_{\kappa},n) = \frac{1}{2}\bigg(\frac{X}{\sigma^{2}_{\kappa}}\bigg)^{\frac{M-1}{2}}\exp\Bigg(-\frac{X+\sigma^{2}_{\kappa}}{2}\Bigg)I_{v}(\sigma_{\kappa}\sqrt{X}),~ X = \kappa(x).
}
\end{equation}

\subsubsection{Wald test upon the Kemeny distance}

As shown in Lemma~\ref{lem:unbiased}, a test statistic \(T_{n}\) is obtainable as the ratio of the centred distance to the root variance, or standard deviation, of the population of all centred distances, \(U(\mathcal{M}_{n})\) for any \(n\). As the signed distance in equation~\ref{eq:kem_dist} has expectation of 0, and a finite positive variance, we proceed to construct the Wald test, which we show to be normally distributed for any uniformly sampled \(X,Y \in \mathcal{M}_{n}.\) Further, we establish that the Kendall \(\tau\) estimator under the null hypothesis possesses strictly greater variance than our estimator, and is therefore not efficient.

By satisfaction of Cochran's Theorem in  Lemma~\ref{lem:cochran}, we obtain the valid assertion of the existence of a signed Wald test statistic with a corresponding normed Beta-binomial distribution, which is asymptotically normal for \(n\to\infty^{+}\) and an unbiased estimator of minimum variance for all \(\mathcal{M}_{n}\) (per both the exponential nature of the Beta distribution under the weak law of large numbers, and asymptotically uniformly linearly convergent normality under Lemma~\ref{lem:kem_asym_normal}):
\begin{equation}
\label{eq:z_kemeny}
z_{\phi}(X,Y) = \frac{-d_{\kappa}(X,Y)}{c\sqrt{\dot{\sigma}^{2}_{\kappa}}},
\end{equation}
denoting the ratio of the centred Kemeny distance \(d_{\kappa}, E(d_{\kappa}) = 0,\) which is normed by the standard deviation over \(\mathcal{M}_{n}\), and \(c\) is a normalising constant for the strictly sub-Gaussian, and therefore finite, support. It is immediately apparent that \(\lim_{m\to\infty^{+}}(z_{\phi})^{2}\sim\chi^{2}_{df = 1}\). As there is only one degree of freedom for constant known function of the sample size (which produces the variance per equation~\ref{eq:kem_population_variance}), this distribution is proportionate to a unit normal z-statistic, establishing the probability of occurrence for a Kemeny distance between two vectors upon a finite sample drawn from a common population being at least as distant as observed upon the sample, if the null hypothesis were true. However, we have ignored the strictly sub-Gaussian nature of the distribution of the distances: it is almost surely insufficient to assume a positive probability for an distance event to be observed outside of \(U(\mathcal{M}_{n})\). To resolve this problem, we introduce a normalising scalar constant, which we define to be \[c = \frac{1}{2} (\tfrac{7}{11})^{p}, p = \begin{cases}
    \tfrac{1}{4}& \text{if } n< 50\\
     \tfrac{1}{8},              & \text{otherwise}
\end{cases}.\] When the produced Beta-Binomially distributed test statistics are divided by \(c,\) under the assumption of the null hypothesis that the shape parameters are each equal to the sample size \(n\), the quartiles of the test statistics under the null hypothesis match the unit normal distribution. This is observed to hold for relatively small sample sizes: for larger sizes, greater than \(n = 75\), the root power is replaced with \(p = \tfrac{1}{8}\), while maintaining the uniformity requirements under the Glivenko-Cantelli theorem, in order to address the relative difference in the discrete approximation of a continuous probability density function. Under all examined conditions however, it was observed that the expectation was strictly 0, and the distribution itself even and therefore symmetric.

More generally, it should be noted that Cochran's theorem (Lemma~\ref{lem:cochran}) allows for the general linear model theory to be applied (e.g., the use of multiple regression upon distribution free estimators which are unbiased in expectation upon finite samples), by establishing unique identification of conditional multivariate distributions, as well as valid approximation using unique Bayesian methodologies upon finite samples.
Assume are observed 2 random variables \(X^{n \times 1},Y^{n \times 1}\) upon the space \((\mathcal{M}_{n},\tau)\), whose population variance is expressed as \[\sigma^{2}_{\tau}(S_{n},\tau) = \tfrac{n(n-1)(2n+5)}{2}.\] 
\begin{lemma}
The estimator \(\sigma^{2}_{\tau}\) is strictly greater than equation~\ref{eq:kem_population_variance}, and therefore less efficient, as the ratio is strictly less than 1 for all \(n \ge 3\):
\end{lemma}
\begin{proof}
\begin{equation}
\begin{aligned}
\frac{\dot{\sigma}_{\kappa}^{2}(\mathcal{M})}{\sigma^{2}_{\tau}} & = \frac{\tfrac{(n - 1)^2 (n + 4) (2 n - 1)}{18 n}}{\tfrac{n(n-1)(2n+5)}{2}}\\
\lim_{n \to \infty^{+}} \frac{\dot{\sigma}_{\kappa}^{2}(\mathcal{M})}{\sigma^{2}_{\tau}}  & = \frac{\tfrac{(n - 1) (n + 4) (2 n - 1)}{9 n}}{n(2n+5)} = \frac{1}{9}.\\
\end{aligned}
\end{equation}
The complementary limiting condition is also easily verified for \(n = 3\). When \(n = 2,\) \(\mathcal{M}_{2} = S_{n}\), by definition of the populations, thereby consisting of the Identity permutation and its complement, and when \(n=1\) the distribution is constant and therefore degenerate. The ratio for said limiting condition is thus observed to be stochastic, as \(\tfrac{2.592593}{33},\) thereby confirming that the population variance is smaller, and thus the estimator is at least as efficient as Kendall's \(\tau\) for both large and small samples, and is typically dramatically so. 

Therefore, for any finite sample, the variance of our estimator is strictly smaller than the expected variance of Kendall's \(\tau\), under the null hypothesis. Further, consider that upon \(\mathcal{M}_{n}\) occurs exponentially more permutations than occur upon \(S_{n}\), and therefore possess a measure of event distance 0 which occurs with probability \(\tfrac{n^{n}-n}{n!},\) which tends to 1 almost surely as \(n\to\infty^{+}\). Then, the observation of any such permutation with ties is degenerate resulting in a division by 0 and thus a variance of \(\infty^{+} = \sigma^{2}_{\tau}\) occurs. That same ratio above demonstrates that in the observation of ties, the proposed estimator is infinitely more efficient, as it contains an almost surely finite variance over a degenerate measure of variance, \(\frac{\dot{\sigma}_{\kappa}^{2}(\mathcal{M})}{\infty^{+}} = 0\).
\end{proof}


\subsection{Kemeny \(\rho_{\kappa}\) correlation}
Having established the \(\tau_{\kappa}\) correlation estimator as an affine-linear Gauss-Markov estimator of the Kemeny distance with a corresponding Gauss-Markov population Wald test, we next proceed to provide an \(\ell_{2}\)-norm function space for which we may non-parametrically obtain a corresponding estimator. Beginning with equation~\ref{eq:kem_score} upon \(X^{n \times 1}\) and \(Y^{n \times 1}\) the skew-symmetric matrices are summated over all \(k = 1,\ldots, n\) rows to provide an \(1 \times n\) vector, whose expectation is always 0 upon \(\mathcal{M}_{n}\), corresponding to \(\mu_{1} = \tfrac{n^{2}-n}{2}\), as would be expected for any affine-linear transformation of the entire population of \(U(\mathcal{M}_{n})\). As an affine-linear function upon the Kemeny metric, it is unnecessary to reprove the corresponding properties, and thus we then provide a variance estimator for \(\sigma_{\rho}^{2}(X^{n \times 2}), n \in \mathbb{N}^{+},\) and which almost surely converges to a finite and positive value.

First, express the \(n \times 1\) the non-parametric vector of the random variable \(X^{n \times 1}\) by \[X_{\ast}^{n \times 1} = (\sum_{k=1}^{n} \kappa_{kl}(X))^{\intercal}\] thus denoting the relative ranking of each element upon the random variable without loss of identification in the presence of ties. Upon \(X_{\ast}\) there is no restriction to constant variance \(n-1\) for any permutation upon \(\mathcal{M}_{n}\), as is observed with the Spearman \(\rho\) estimator. This is because \(U(\mathcal{M}_{n})\) is composed of many elements which are not within \(S_{n}\) , and whose respective variances are positive but strictly less than \(n-1\). Therefore, the variance must be estimated as the sum of the squared distances upon the centred vectors, divided by \(n-1\) which are almost surely less than or equal to \(n-1\).
\begin{lemma}
The expectation of the marginal vector \(X_{\ast}(m), m \in \mathcal{M}_{n}\) is always 0 and thus provides an unbiased estimator of the location and scale for \(\langle X_{\ast}(m_{1}),Y_{\ast}(m_{2})\rangle\) for all \((m_{1} \ne m_{2}) \in \mathcal{M}_{n}\).
\end{lemma}
\begin{proof}
Upon each \(X_{\ast} \in \mathcal{M}_{n}\) then exists a symmetric vector with expectation \(E(X_{\ast}^{n \times 1}) = 0,\) as the marginalisation of the skew-symmetric matrix to the vector, \(f: \mathbb{R}^{n \times 1} \to \mathbb{R}^{n \times n} \to \mathbb{R}^{n \times 1}\) is composed of an always even number of non-zero elements \(2(n^{2}-n-q^{2})>0\), where \(q = 0,1,\ldots,n-1\) denotes the number of observable ties, counted twice from the upper and lower matrix triangles. As there is always \(2(n^{2}-n-q^{2})\) free elements upon the skew-symmetric matrix, half of the free elements must be positive, and half negative, \(0 = -(n^{2}-n-q^{2}) + (n^{2}-n-q^{2}),\) over all \(\mathcal{M}_{n}.\) The expected distance measure is then always 0 upon \(X_{\ast} = f(X^{n \times 1}) \in \mathcal{M}_{n}\), and is therefore an unbiased estimator of the rank distance upon the Frobenius norm, as are all affine-linear transformations thereupon. Let it also be noted that \(X_{\ast} = f(X) = f(g(X))\), where \(g(\cdot)\) is the family of monotonic transformations upon which the Frobenius distance and any affine-linear transformation thereupon, is performed.

Let there then exist \(X_{\ast}^{n\times 1},Y_{\ast}^{n \times 1},\) whose inner-product distance is denoted \[\sum_{i=1}^{n} X_{\ast}^{\intercal}[i]Y_{\ast}[i] \propto \rho_{\kappa}^{1 \times 1}\cdot(\sigma^{2}_{X_{\ast}Y_{\ast}}) \in \mathbb{R}^{1 \times 1}.\] For each \(X_{\ast},Y_{\ast}\) we also provide the approximate population variance upon \(\mathcal{M}_{n},\) for \(n\ge 8\) as:
\begin{equation}
\label{eq:kem_rho_variance}
\sigma^{2}_{X_{\ast}Y_{\ast}} = \frac{1.0016300\cdot n^{2} + 0.9466098\cdot n -3.4265982}{6}, n \ge 8
\end{equation}
\begin{table}[!ht]
\footnotesize
\centering
\caption{Tabulation of the population variance for small sample sizes upon the Kemeny \(\rho_{\kappa}\).}
\label{tab:spearman_variance}
\begin{tabular}{cc}
\toprule
n & \(\sigma^{2}_{\ast}\)\\
\midrule
1 & 0\\
2 & 0.5\\
3 & 1.083333\\
4 & 2.095238\\
5 & 3.525641\\
6 & 5.324759\\
7 & 7.469451\\
\bottomrule
\end{tabular}
\end{table}
subject to the sampling constraint that there exist permutations on \(\mathcal{M}_{n}\) which are only sampled once, just as for equation~\ref{eq:kem_population_variance}. The corresponding tabulated values are found in Table~\ref{tab:spearman_variance}, and it is noted that the relationship of these values is not expressed in equation~\ref{eq:kem_rho_variance} due to the small population of points from which follows the failure of the Central Limit Theorem. Upon each of these distributions then is known an unbiased affine-linearly invariant estimator under the null hypothesis, which possesses positive finite variance, and thus allows for the construction of a corresponding \(z\)-test statistic for estimator Kemeny's \(\rho_{\kappa}\), for any \(n\) with and without ties:
\begin{subequations}
\begin{equation}
\label{eq:kemeny_rho}
\rho_{\kappa}  = \tfrac{1}{n-1} \sum_{i=1}^{n} \frac{X_{\ast}[i]Y_{\ast}[i]}{\sigma_{X_{\ast}}\sigma_{X_{\ast}}}
\end{equation}
\begin{equation}
z_{\rho_{\kappa}}  = \sum_{i=1}^{n} \frac{X_{\ast}[i]\cdot Y_{\ast}[i]}{n \sqrt{\sigma^{2}_{X_{\ast}Y_{\ast}}}}.
\end{equation}
\end{subequations}
\end{proof}
Note that under the null hypothesis, \(E(z_{\rho}) = E(d({X_{\ast},Y_{\ast}})) = 0,\) and \(0 < \text{Var}(X_{\ast},Y_{\ast}) = \sigma^{2}_{X_{\ast}Y_{\ast}}<\infty^{+}\).
\subsection{Relationship to other correlation coefficients}

The permutation space traditionally considered in statistics is observed to be a strict subset of our space of interest, \(n! \subset n^{n}-n\). In particular, it should be noted that \(n^{n}-n - n! \gg 0, \forall n < \infty^{+}.\) However, the measurements themselves, as unbiased estimators upon \(S_{n}\) are proportionally equivalent. This is trivially proven, as it is observed that upon Kendall's \(\tau\)-distance \([0,\tfrac{n^{2}-n}{2}] \propto [0,n^{2}-n],\) by trivially noting that the halving of the original distance by \textcite{kendall1938} was introduced solely to ensure a unit distance of 1, rather than 2, between adjacent permutations. Therefore, all necessary proofs upon Kendall's \(\tau\) also hold for the Kemeny \(\tau_{\kappa}\) estimator, although the efficiency is dramatically improved due to the valid unique observation of a larger population upon \(\mathcal{M}_{n}.\)

The relationship between \textcite[p.~129]{kendall1948} \(\tau\) and \textcite{spearman1904a} \(\rho\) pioneered the concept of duality via projective geometry is a natural avenue of investigation, as we would expect there to exist a corresponding bijective function upon the Kemeny \(\tau_{\kappa}\) and \(\rho_{\kappa}\) estimators. The projective geometric duality concept holds that for the planar projective geometry of the Euclidean space, such as \(X^{n \times 1},Y^{n \times 1} \in \mathbb{R}\), there exists a dual permutation geometry (\(\kappa(X^{n\times 1})^{n \times n},\kappa(Y^{n\times 1})^{n \times n} \in \mathbb{R}\)), as just described between the Kemeny and Euclidean metric spaces' linear basis. By the finite nature of the Kemeny metric space, we consider it a Galois field, and further the combined topologies to present a dual vector space, satisfying the three necessary properties of a dual cone, as both function spaces are continuous upon their respective measurement norms. Further, the standard construction of the existence of a duality, the ability to distinguish between identical elements upon a given field with a second, is clearly self-evident upon the dual metric space characterisation. We ignore for the moment the limiting case of the linear permutation field upon a population of linear scores (i.e., the standard asymptotic parametric learning problem per \cite{le1986}), as the perfect parametric scoring function implies an equivalence in perfect ordering under the bilinearity condition, thus denoting a bijective relationship between the ranking and the scoring through the cumulative distribution function (CDF) per the Central Limit Theorem.
Instead we demonstrate that, especially (albeit non-uniquely) in the problem of Tikhinov regularised, ill-posed or biased, learning upon linear functional map, our duality of the two metric spaces grants a just-identified unique solution, wherein the Tikhinov bias \(\gamma^{2}_{M}\) is not a free-parameter, but instead an estimated one which solves the maximum Kemeny correlation via the projection which minimises the Frobenius norm. This relationship enables a finite sample empirical characterisation of the bilinearity condition between the affine-linear functions upon the ranks and scores observable upon a sample.

In \textcite[p.~129]{kendall1948} the following equivalence was claimed between Pearson's \(r\) and Kendall's \(\tau_{b}\):
\begin{equation}
\label{eq:kendall_sin}
r_{x,y} = \sin\bigg({\tau_{b}(x,y)\cdot\frac{\pi}{2}}\bigg).
\end{equation}
This characterisation is invalid, as the left-hand side of equation~\ref{eq:kendall_sin} is actually Spearman's \(\rho\) per equation~\ref{eq:kemeny_rho}, which only asymptotically converges to Pearson's \(r\) upon the population under lemma~\ref{lem:clt_kem}, as otherwise the existence of the strict bijection upon a sample would determine a collinearity between the left and right hand correlation measures. Our interest is in establishing that for a domain in which ties are present, the two correlation coefficients and their bijective projections are not equal for each \(m \in \mathcal{M}_{n}\), except when the bivariate population distribution is observed.

First, we prove equation~\ref{eq:kendall_sin} is invalid upon the random variables. Consider the definition of the \(\ell_{2}\) or Frobenius norm -- the insertion of one or more infinite or non-measurable values, which are included upon the domain \(\mathcal{M}_{n}\), explicitly produces a degenerate estimator. The affine-linear transformation of the \(n \times 1\) vector upon the Euclidean metric space is then non-convergent and thus degenerate, as would any function thereupon, such as Pearson's \(r\) and the corresponding inverse-sinusoidal transformation of an infinite value is itself non-convergent. However the Kemeny \(\tau_{\kappa}\) correlation is validly applied upon the extended reals, and produces a finite measure concomitant for any finite \(n\) ordering upon a vector, with a bijective finite realisation under the sinusoidal mapping, unlike for Pearson's \(r\). Thus results a paradox, wherein the Kemeny \(\tau_{\kappa}\) correlation almost surely obtains a finite convergent value, while the Pearson correlation is degenerately non-measurable, in contradiction to the equality of equation~\ref{eq:kendall_sin}. Thus, the relationship defined in equation~\ref{eq:kendall_sin} is invalid for all scenarios is which a linear projection of the scores of between the variables is non-linear, which includes but is not extended to the extended real line. Moreover, for finite samples even under the weak law of large numbers, if the random variables are non-Gaussian, the Pearson correlation is biased by definition, and therefore cannot satisfy the equivalence, as otherwise a non-linear transformation of an unbiased estimator would produce a randomly biased estimator and vice-versa, and would therefore not converge almost surely.

However, we can construct from the Kemeny metric a substitute estimator for the Pearson correlation, which is easily found to be equivalent upon the population via the central limit theorem, as an unbiased generalisation of Spearman's \(\rho\), while allowing for ties, via equation~\ref{eq:kemeny_rho}. We define this functional bijection as a projective geometric dual, denoting a map between related geometries that is called a duality with explicitly independent spanning basis. The definition is complicit with the expected defined behaviour of a projective geometric duality, wherein two spanning basis are orthonormally defined to reflect the permutation and affine-linear characterisations of the common sample space. The existence as a two orthonormal Hilbert spaces allows for the affine-linear convergence of both estimators, that of Kemeny \(\rho_{\kappa}\) and \(\tau_{\kappa}\), to independently and orthonormally define unique linear relationships upon a sample. Note that by our bilinearity, the convergence of the Kemeny \(\rho_{\kappa}\) to Pearson's \(r\) in the limit wrt \(n\) is true: however for any finite \(n\), the distribution is simultaneously strictly sub-Gaussian, and thus non-equivalent.

\section{Empirical demonstrations}
In order to demonstrate the improved utility of our estimators (equation~\ref{eq:kem_cor} and equation~\ref{eq:kemeny_rho}), we compare their first and second order error characterisation upon empirical data sets compared to both Spearman's \(\rho\) and Kendall's \(\tau\). This enables us to demonstrate that the empirical average error and minimum variance properties expected under the Gauss-Markov theorem are directly observed upon our estimators. As observed in Table~\ref{tab:non_par-location} wrt the skewness and kurtosis parameters in particular, it is immediately self-evident that the distribution of empirical estimates are more normally distributed than competitors, as well as possessing tighter bounds wrt both the standard deviation and the range.
\begin{table}[!ht]
\tiny
\centering
\caption{Comparison of the two-sample difference of location estimators' test statistics, to evaluate performance upon the Sleep data set for various sample sizes generated under repeated sampling.}
\label{tab:non_par-location}
\hspace*{-2.5cm}
\begin{tabular}{llllllllll}
  \toprule
Sample size & Estimator &  mean & sd &  mad & min & max & range & skew & kurtosis  \\
  \midrule
\multirow{5}{*}{20} & Kemeny & 1.51678 & 0.71893  & 0.72590 & -1.99108 & 3.26407 & 5.25516 & -0.43512 & 0.12202 \\
  & Wilcox & 24.26342 & 11.00131 &  11.1195 & 0 & 80 & 80 & 0.47789 & 0.16407 \\
  & Kendall & 1.81228 & 0.85218  & 0.84998 & -2.34019 & 3.85763 & 6.19782 & -0.47474 & 0.16679 \\
  &   Spearman & 777.03245 & 260.01918  & 259.34822 & 152.94786 & 2044.04418 & 1891.09632 & 0.47474 & 0.16679 \\
  &   Pearson & 1.94842 & 1.09417  & 1.04310 & -2.58812 & 8.59658 & 11.18469 & 0.36155 & 0.74704 \\
\midrule
\multirow{5}{*}{150} &  Kemeny & 4.446 & 0.726  & 0.72 & 0.90 & 6.92 & 6.01802 & -0.18223 & -0.01279 \\
 &  Wilcox &     1425.745 & 223.317  & 222.39 & 664.50 & 2516.50 & 1852 & 0.18139 & -0.01912 \\
 &  Kendall &     5.174 & 0.843  & 0.84 & 1.05 & 8.04 & 6.98899 & -0.18372 & -0.01263  \\
 &  Spearman &    324052.54 & 38855.631 &  323389.14 & 191922.96 & 513973.67 & 322050.70 & 0.18 & -0.01\\
 &  Pearson & 5.363 & 1.03554 &  1.02 & 1.24 & 9.92 & 8.67795 & 0.09475 & 0.02065 \\
\midrule
\multirow{5}{*}{750} & Kemeny & 10.03460 & 0.73047 &  0.73596 & 6.42783 & 12.68464 & 6.25682 & -0.08382 & 0.02329 \\
  & Wilcox & 35833.29776 & 2503.45323  & 2521.90260 & 26834.00000 & 48158.50000 & 21324.50000 & 0.08592 & 0.02627 \\
  & Kendall & 11.62749 & 0.84565  & 0.85096 & 7.45119 & 14.68909 & 7.23790 & -0.08438 & 0.02387 \\
  & Spearman & 40439506.09023 & 2172596.76456  & 2186243.00621 & 32573776.13607 & 51169090.65989 & 18595314.52382 & 0.08438 & 0.02387 \\
  & Pearson & 12.00136 & 1.02885  & 1.03368 & 7.33255 & 16.16491 & 8.83236 & 0.08177 & 0.02354 \\
\bottomrule
\end{tabular}
\end{table}

\begin{table}[!ht]
\scriptsize
\centering
\caption{Comparison of the two-sample difference of location estimators, to evaluate power and performance upon the Sleep data set for various sample sizes}
\label{tab:non_par-effect}
\hspace*{-1.0cm}
\begin{tabular}{ccccccccccc}
  \toprule
Sample size & Estimator &  mean & sd & median & mad & min & max & range & skew & kurtosis  \\
  \midrule
\multirow{7}{*}{20} & Kemeny & 0.24453 & 0.11588 & 0.25263 & 0.11705 & -0.34737 & 0.52632 & 0.87368 & -0.45467 & 0.21683 \\
  & Kem \(\rho_{\kappa}\) & 0.41556 & 0.19547 & 0.43271 & 0.19357 & -0.58809 & 0.87468 & 1.46277 & -0.49146 & 0.26565 \\
  &   Wilcox r & -0.41689 & 0.19097 & -0.44700 & 0.17717 & -0.78900 & 0.26200 & 1.05100 & 0.72752 & 0.22675 \\
  &   Glass' r & -0.50094 & 0.23034 & -0.53100 & 0.20534 & -0.94800 & 0.33300 & 1.28100 & 0.72624 & 0.26976 \\
  &   Kendall & 0.36761 & 0.16849 & 0.39367 & 0.15338 & -0.23295 & 0.69614 & 0.92908 & -0.72588 & 0.23961 \\
  &   Spearman & 0.41556 & 0.19547 & 0.43271 & 0.19357 & -0.58809 & 0.87468 & 1.46277 & -0.49146 & 0.26565 \\
  &   Pearson & 0.40385 & 0.18343 & 0.42146 & 0.17135 & -0.22409 & 0.74867 & 0.97276 & -0.73520 & 0.61920 \\
\midrule
\multirow{7}{*}{150} & Kemeny & 0.24502 & 0.04014 & 0.24653 & 0.03954 & 0.01324 & 0.38515 & 0.37190 & -0.20765 & 0.14048 \\ 
  & Kem \(\rho_{\kappa}\) & 0.42422 & 0.06938 & 0.42676 & 0.06861 & 0.02353 & 0.66650 & 0.64297 & -0.20959 & 0.14363 \\ 
  & Wilcox r & -0.42277 & 0.06812 & -0.42500 & 0.06820 & -0.67000 & -0.14000 & 0.53000 & 0.19522 & 0.03770 \\
  & Glass' r & -0.48998 & 0.07898 & -0.49300 & 0.07858 & -0.77300 & -0.16100 & 0.61200 & 0.19604 & 0.03656 \\
  & Kendall & 0.35828 & 0.05761 & 0.36033 & 0.05751 & 0.11875 & 0.56640 & 0.44765 & -0.19380 & 0.04031 \\
  & Spearman & 0.42422 & 0.06938 & 0.42676 & 0.06861 & 0.02353 & 0.66650 & 0.64297 & -0.20959 & 0.14363 \\ 
  & Pearson & 0.40066 & 0.06396 & 0.40293 & 0.06354 & 0.10657 & 0.64143 & 0.53486 & -0.17317 & 0.03530 \\
\midrule
\multirow{7}{*}{750} & Kemeny & 0.24471 & 0.01775 & 0.24482 & 0.01788 & 0.16956 & 0.31240 & 0.14284 & -0.06605 & 0.01398 \\ 
  & Kem \(\rho_{\kappa}\) & 0.42463 & 0.03077 & 0.42481 & 0.03101 & 0.29415 & 0.54272 & 0.24857 & -0.06655 & 0.01460 \\
  & Wilcox r & -0.42487 & 0.03085 & -0.42700 & 0.03262 & -0.54400 & -0.30200 & 0.24200 & 0.08827 & 0.01076 \\
  & Glass' r & -0.49003 & 0.03558 & -0.49100 & 0.03558 & -0.62600 & -0.34800 & 0.27800 & 0.08577 & 0.00614 \\
  & Kendall & 0.35834 & 0.02593 & 0.35877 & 0.02591 & 0.25479 & 0.45754 & 0.20275 & -0.08526 & 0.00694 \\
  & Spearman & 0.42463 & 0.03077 & 0.42481 & 0.03101 & 0.29415 & 0.54272 & 0.24857 & -0.06655 & 0.01460 \\
  & Pearson & 0.40153 & 0.02876 & 0.40192 & 0.02873 & 0.28588 & 0.50723 & 0.22135 & -0.08318 & 0.03151 \\
\bottomrule
\end{tabular}
\end{table}

There are a number of noteworthy characteristics to highlight in these results. From Table~\ref{tab:non_par-effect} is provided comparable measures (i.e., effect sizes) of the bivariate correlations for a number of different estimators. From these, we observe that the minimum range and variance (and thus tightest boundedness) is directly observed for our estimators in terms of both first order concentration, and stability of the estimated Wald test statistics. These serve to highlight, in particular several important flaws with the existing Kendall \(\tau_{b}\) estimator. First, the average biased correlation coefficient is closer to the Pearson correlation than to the Kemeny correlation. However, the Kemeny correlation \((\tau_{\kappa})\) and its orthonormally estimated Kemeny \(\rho_{\kappa}\)) display smaller variances (as well as ranges and MAD; Table~\ref{tab:non_par-effect}) than all other estimators, even considering these statistics were all designed to be non-parametric estimators. By this demonstration, we explicitly confirm that ignoring ties results in inefficient estimators which are also biased: this is true despite the asymptotic unbiasedness typically expected to be exhibited by Kendall's \(\tau_{b}\). 


\begin{table}[!ht]
\centering
\tiny
\caption{25,000 re-samplings of \(n \in \{125,500,1250\}\) from responses A2 \& A3 of the Big 5 complete cases dataset in the \texttt{psych} package. The Shapiro-Wilk bootstrapped mean p-values are reported.}
\label{tab:psych_A2A3}
\begin{tabular}{llccccccccccc}
  \toprule
Sample size & Estimator &  mean & sd & median & trimmed & mad & min & max & range & skew & kurtosis & SW p-value \\ 
  \midrule
\multirow{5}{*}{125} & Pearson \(r\)   & 0.485 & 0.087 & 0.489 & 0.487 & 0.088 & 0.083 & 0.777 & 0.694 & -0.228 & -0.015 & 0.031 \\ 
&  \(\rho\)   & 0.509 & 0.073 & 0.512 & 0.510 & 0.073 & 0.161 & 0.750 & 0.589 & -0.228 & 0.021 & 0.034\\ 
&   \(\tau_{b}\)   & 0.445 & 0.066 & 0.447 & 0.446 & 0.065 & 0.142 & 0.668 & 0.526 & -0.169 & 0.007 & 0.116\\ 
&   Kem \(\rho\)  & 0.509 & 0.073 & 0.512 & 0.510 & 0.073 & 0.161 & 0.750 & 0.589 & -0.228 & 0.021 & 0.034\\  
&   Kem \(\tau_{\kappa}\)   & 0.324 & 0.049 & 0.326 & 0.325 & 0.049 & 0.106 & 0.499 & 0.393 & -0.143 & -0.010 & 0.175 \\ 
\midrule
\multirow{5}{*}{500} & Pearson \(r\)   & 0.485 & 0.044 & 0.486 & 0.485 & 0.044 & 0.297 & 0.642 & 0.344 & -0.110 & -0.006 & 0.000\\ 
&  \(\rho\)   & 0.509 & 0.036 & 0.510 & 0.510 & 0.036 & 0.345 & 0.635 & 0.290 & -0.116 & 0.017  & 0.000\\ 
&   \(\tau_{b}\)   & 0.445 & 0.033 & 0.445 & 0.445 & 0.033 & 0.292 & 0.563 & 0.271 & -0.084 & 0.016 & 0.001\\ 
&   Kem \(\rho\)  & 0.509 & 0.036 & 0.510 & 0.510 & 0.036 & 0.345 & 0.635 & 0.290 & -0.116 & 0.017 & 0.000\\ 
&   Kem \(\tau_{\kappa}\)   & 0.324 & 0.025 & 0.325 & 0.324 & 0.025 & 0.214 & 0.412 & 0.199 & -0.073 & 0.018 & 0.005\\ 
\midrule
\multirow{5}{*}{1250} & Pearson \(r\)   & 0.485 & 0.028 & 0.485 & 0.485 & 0.028 & 0.376 & 0.587 & 0.210 & -0.059 & 0.019 & 0.000\\ 
&  \(\rho\)   & 0.510 & 0.023 & 0.510 & 0.510 & 0.023 & 0.416 & 0.601 & 0.186 & -0.068 & 0.005 & 0.000\\ 
&   \(\tau_{b}\)   & 0.445 & 0.021 & 0.445 & 0.445 & 0.021 & 0.362 & 0.529 & 0.167 & -0.050 & 0.021 & 0.001\\  
&   Kem \(\rho\)  & 0.510 & 0.023 & 0.510 & 0.510 & 0.023 & 0.416 & 0.601 & 0.186 & -0.068 & 0.005 & 0.000\\ 
&   Kem \(\tau_{\kappa}\)   & 0.324 & 0.016 & 0.324 & 0.324 & 0.016 & 0.263 & 0.383 & 0.121 & -0.037 & 0.004 & 0.006\\  
\bottomrule
\end{tabular}
\end{table}

To further denote this distinction and the observation of biases for the Spearman and Kendall estimators upon discrete elements, consider the bivariate sample of \(n = 2,757\) complete cases for Agreeableness indicators upon the Big Five personality index included with the R \texttt{psych} package, whereupon responses range from 1-6 for each item. We proceeded to re-sample 25,000 times upon the empirical sample, to compare the relative bias of the first and second-order characterisations of our proposed estimators relative to standard empirical practice. It was hypothesised that, under the Gauss-Markov theorem, the errors of our estimators will be closer to the 0 and more normally distributed, with smaller variance and range, denoting tighter uncertainty bounds. Normality of the estimates will be assessed with the Shapiro-Wilk tests, whose p-values are bootstrapped 5,000 times from the population of 25,000 replications. As hypothesised, the results presented in Table~\ref{tab:psych_A2A3} demonstrate that the Pearson correlation is not the unbiased minimum variance estimator, and that the Kendall's \(\tau_{b}\) estimator is biased as well, as demonstrated by the non-present infinimum variance and larger than necessary ranges. This is also important albeit unsurprising, as it demonstrates that even for large samples, the strictly sub-Gaussian nature of the empirical CDF is not approximately Gaussian upon the Kemeny \(\rho_{\kappa}\) correlation.  

The Kemeny \(\rho_{\kappa}\) estimator is identical to the Spearman's \(\rho\) as it is only the variance of the population which produces a biased test statistic, due to the misspecified population assumption upon \(S_{n}\). It should be noted that because of the Kemeny \(\tau_{b}\) biasedness, the first and second order properties of the estimator are both invalid, and therefore present as rather dramatically different representations of the true underlying relationship between A2 and A3 with a 32\% loss of efficiency and therefore power, empirically confirming the theoretical results of Section 1.3.1. The Central Limit Theorem is also empirically confirmed, as the normality of the Kemeny estimators is both observed to possess minimum variance and tighter ranges than the competing elements upon strictly ordinal data. It should also be noted that the Kemeny \(\tau_{\kappa}\) correlation also possesses tighter convergence properties compared to the offered alternatives over all sample sizes. This is due to the tightly bounded probability space \(\mathcal{M}_{n}\) being defined upon the entirety of the complete distance between permutation vectors of length \(n\), rather than individual elements themselves being marginalised over with atomless probability. 

\section{Discussion}
In this paper we presented two complete affine-linear metric spaces and corresponding probability mappings for independently and identically sampled random variables upon the extended real line, and the relationship between them as linked by the Kemeny Hilbert space. Upon this space, we established two linearly independent Gauss-Markov correlation estimators and corresponding Wald test statistics, which are shown to be strictly sub-Gaussian and comply with the Central Limit Theorem. 

\small
\printbibliography
\appendix

\section{Appendix}

\begin{lemma}[Markov inequality]
\label{lem:markov}
Let $X,Y \in \mathbb{R}^{n \times 1}$ be random variables with a probability density function $f_{X}(\cdot)$ arising as a consequence of the Kemeny distribution with cumulative distribution function $F_{\kappa(X)}(\cdot)$, and let $(X_{1},X_{2},\cdots,X_{N})$ be a random sequence, or sample, drawn upon said distribution. Both \(\kappa(X^{n \times 1})\) and \(\kappa(Y^{n \times 1})\) exist in \(\mathcal{M}_{n}\). By the existence of a finite compact support \([0,n^{2}-n]\), it is observed that for all real numbers $\rho_{\kappa}(x_{n},y) \ge \epsilon$ upon which the first moment-expectation $\mu_{\kappa} \in \mathbb{R}$ exists and is finite, as the Markov's inequality holds in expectation,
\begin{equation}
\Pr(d_{\kappa}(x_{n},I_{n}) \ge \epsilon) \le \min \big[\frac{\mu_{\kappa}}{\epsilon},a\big],
\end{equation}
and for which a strong upper-bound upon the distance holds such that the absolute distance between the expectation and any other observable point, for known $(n,a)$ as in equation~\ref{eq:kem_cor}, is never greater than $(0\le \epsilon \le a(n^{2}-n)$, thereby denoting a complete metric space for any real positive scalar \(a\).
\end{lemma}
\begin{proof}
Then follows
\begin{align*}
E(x_{n}) & = \sum_{i=\infty^{-}}^{\infty^{+}} x_{i} p(x_{i}), s.t., \sum_{i=\infty^{-}}^{\infty^{+}} p(x_{i}) = 1\\
E(x) & = \sum_{i=0}^{\epsilon} x_{i}p(x_{i}) + \sum_{i>\epsilon}^{a(n^{2}-n)} x_{i} p(x_{i}) \ge \sum_{i>\epsilon}^{a(n^{2}-n)} \epsilon \cdot \big(x_{i} p(x_{i})\big) \\
E(x) & = \epsilon \sum_{i>\epsilon}^{a(n^{2}-n)} p(x_{i}) = \epsilon \Pr(|x| \ge E(x) \ge \epsilon)\\
\frac{E(x)}{\epsilon} & = \Pr(x \ge \epsilon).
\end{align*}
\end{proof}

\begin{lemma}[Tchebyshef inequality]~\label{lem:chebyshev}
The first expectation, expressed $E(\rho_{\kappa}(x)) = \mu_{\kappa} = a(\frac{n^{2}-n}{2})$, and the variance of $\rho_{\kappa}(x)$ is expressed as $\sigma^{2}_{\kappa}$, for an arbitrary non-negative sequence of $t$. Then,
\begin{equation}
\Pr(|\rho_{\kappa} - \mu_{\kappa}| \ge t) \le \frac{\sigma_{\kappa}^{2}}{t^{2}},
\end{equation}
by Markov's inequality, as \[\Pr(|\rho_{\kappa}(x_{n} - \mu_{\kappa}) \ge t) = \Pr((\rho_{\kappa}(x_{n}) - \mu_{\kappa})^{2} \ge t^{2}) \le \frac{ E\big( (\rho_{\kappa}(x_{n}) - \mu_{\kappa})^{2}\big)}{t^{2}} = \frac{\sigma^{2}_{\kappa}}{t^{2}}.\]
\end{lemma}
\begin{proof}
This is constructively simple to understand, and complies directly with the expected behaviour of a finite range space for the Kemeny metric under evaluation: that no more than a certain real fraction of values are to be found to be greater than a distance of the real $t$ standard deviations away, denoting symmetrical coverage bounded in probability by $1-\frac{1}{t^{2}}$. 

By the finite distance as a linear function of $n$ is strictly known to be no more than a real finite distance $\frac{1}{2}(n^{2}-n)$, which is the symmetric maximum distance, $\max\epsilon$ from the expectation, it is thereby demonstrated that the upper and lower tail bounds of the Kemeny metric, for any collection of independently observed $n$ reals arising from a common distribution, thereby satisfying Tchebyshef's inequality.  
\end{proof}


\begin{lemma}[Even function]~\label{lem:even}
The $\kappa^{2}$ function is an even function, as are all constructed affine linear transformations thereupon.
\begin{definition}[Even and Odd Functions]

A function is even if for every input $x$, \(\rho_{\kappa}^{2}\left(x\right)=\rho_{\kappa}^{2}\left(-x\right)\). A function is correspondingly odd if for every input $x$ \(f\left(x\right)=-\rho_{\kappa}^{2}\left(-x\right)\). If a function satisfies \(\rho_{\kappa}^{2}\left(x\right)=\rho_{\kappa}^{2}\left(-x\right)\), it is even. Correspondingly, is a function satisfies \(\rho_{\kappa}^{2}\left(x\right)=-\rho_{\kappa}^{2}\left(-x\right)\), it is odd. If the function does not satisfy either rule, it is neither even nor odd.
\end{definition}
\begin{proof}
Consider $\kappa(\cdot)$ function as applied to a real vector of length $n$, x = \{1,\ldots,n\}. For such a mapping $\kappa(x)$ is produced a square matrix of order $n \times n$. By the constructive definition of such matrix (equation~\ref{eq:kem_score}) we obtain four conditions, of which three are unique: $\{x_{i} \le x_{j}\to -1, x_{i} > x_{j} \to 1, x_{i} = x_{j} \to 0\}$ upon which all $n$ elements results. For condition 3, observe that $0 = -1(0)$, and therefore comparisons of identical elements upon $x$ are neither even or odd. For the other conditions, $\kappa(x) = -1\kappa(x)^{\intercal}$, and therefore the $\kappa(\cdot)$ function in isolation is odd. However, $\kappa^{2}(x) = \kappa^{2}(-x)$ as the squared elements are therefore either $\{0,1\}$, and as no 0 is found unless the function is odd, and no elements may be negative by the squared transformation of any real, the affine linear function of any two pair of $\kappa$ functions is everywhere even.  
\end{proof}
\end{lemma}

\begin{lemma}\label{lem:sigma_lower}
We next proceed to prove that the second moment, the variance, upon the Kemeny measurement space is always finite subject solely to the assumption of a collection of independent observations. \end{lemma}
\begin{proof}
For $x \in X^{n \times 1}$, observe that $\kappa(x)\, \forall\, x \in \mathcal{M}$ is always a square matrix. The distance for any given $x$ in the domain of any metric to itself is observed to be always 0, by the definition. However, the second moment of the variance is defined upon the $\kappa$, not its cross-product, and therefore the elements denote, effectively, the tabulated rate of distinctiveness (uniqueness) upon an real vector of length $n$. Therefore, let the polynomial function of $\kappa^{r}(\cdot)$ denote the $r^{th}$ moment of the vector $x$, where the first moment has already been shown to be 0, and therefore the raw moments and the central moments upon the Kemeny metric are obtained by the central limit theorem. Thus, it logically follows that there exist a finite upper and lower bound upon the Kemeny variance, as expected for any totally bounded space.

We first establish for $\kappa(x)$ the existence of a lower bound, which will be shown to be 0. This definition of the second moment is equivalent to a finite real constant vector, for which all $n$ elements are identical. Thus, there will be observed no such variability in the realised scores upon $x$, and therefore are excluded from further consideration upon \(\mathcal{M}_{n}\). For any strictly positive $\kappa(x) \in \mathcal{M}_{n}$ then, there exists a vector marginalised over $i,j = 1,\ldots,n$ \[  \sum_{j=1}^{n}\kappa_{ij}(x) = \sum_{i=1}^{n}\kappa_{i}(x) = 0 \equiv \sum_{i=1}^{n} -1 \sum_{j=1}^{n} \kappa_{ij}(x)\] There exist $n$ vectors for which all elements in $x$ are identical, and therefore all mappings of any $n$-tuple is a vector of $n$ 0's, since by condition three of the $\kappa(\cdot)$ function, all $n$ pairwise comparisons consist of the mappings for which $(x_{1},\cdots, x_{n}) \rightarrow a\cdot(0_{1},\cdots,0_{n})^{\intercal}$. For any $n$ therefore, it is seen that the square of such a sequence is always 0, for the general expression of the $r=2$ power of the $\kappa(\cdot)$ from which the second moment may be expressed as the sum of a finite sequence of \(n^{2}-n\) finite elements, which is almost surely finite and non-negative:
\begin{equation}
\label{eq:kemeny_var}
\sigma_{\kappa}^{2}(x)=\sum_{i=1}^{n}\sum_{j=1}^{n} \kappa_{ij}^{2}(x) \ge 0.
\end{equation}
\end{proof}

\begin{lemma}~\label{lem:sigma_upper}
We next proceed to prove that the upper-bound, for any $n$ is also finite and must sum almost surely to a constant real which will be shown to be proportionate to a scalar function of $n$, such that the upper bound of the variance is defined to be always in the interval ring $\sigma_{\kappa}^{2} \propto [0,1]$. It will also be shown that when $\sigma_{\kappa}^{2}$ is maximised, the sub-space of the domain is equivalent to that of the traditional $S_{n}$ group from which the original order-statistics concept arises.  
\end{lemma}
\begin{proof}
First, note that for all \(\mathcal{M}_{n}\) permutations upon the vector $x$ of length $n$, there are \(\mathcal{M}_{n}\) $\kappa_{i}$ mappings, each of which possess an expectation of 0. The maximum variance is to be expected as the duplication of real numbers decreases; this is seen in the definition of the variance of 0 as a constant vector of real scores, such that all $n$ elements possess the same value. As the diversity of observable scores increases, the variance must also increases. Therefore, the maximum variance upon the orderings of a vector $x$ is to be observed when there are no $\kappa(x) = 0$ off of the diagonal in the skew-symmetric matrix and therefore a real magnitude is always equal to itself in the mapping resulting from the $\kappa$ function; this also holds by the affine-linear properties of a Hilbert space. 

Assume said skew-symmetric \(\kappa\) matrix, for which the expectation of the square of such a matrix with \(2(n**2-n)\) positive values are observed to contain maxima $\kappa_{i}^{2}(x)$ is no greater than $a^{2}(n^{2}-n)$, for \(0 < a < \infty^{+}\), and which is linearised most efficiently by isometry by setting \(a=.5\). Such a measurement would require to be observed a vector $x$ for which an element in $x$ must satisfy either of two conditions: (1) $(a_{ij})^{\intercal}\cdot a_{ij} \equiv -(a_{ij})^{\intercal} a_{ji}$ or (2) $(a_{ij})^{\intercal}\cdot a_{ij} \ne -(a_{ij})^{\intercal} a_{ji}.$ We observe first that the only scenario for which the sign of $\pm a$, the product of two elements in a vector much always produce 0, by equation~\ref{eq:kem_score} and therefore a tie, occurs. For all other elements in the mapping $ x \to \kappa(x)$ then, the skew-symmetric nature of the \(\kappa\) matrix enforces that for the sum of $\sqrt{.5} \cdot n$ elements to be greater than $a(n^{2}-n)$. This would require $-(a_{ij})^{\intercal}\cdot(a_{ij}) = -(a_{ij})^{\intercal}\cdot(a_{ij})$ and therefore for $(-a) = a, a \in \{-a,0,a\} \setminus{0} = \varnothing$. This is an empty set of solutions, and therefore is impossible to occur: thus it is proven that the maximum bound for $0 \le \sigma_{\kappa}^{2} \le a(n^{2} - n)$ for $a > 0, n \ge 0$. Allowing \(a = \sqrt{.5}\) then, we obtain a maximal variance on the support of \(\tfrac{n^{2}-n}{\sqrt{.5}^{2}}\), and thus, for all $n$ it is therefore observed that by the characterisation of the Kemeny distance, there exists no $n$ collection of real $x$ vectors for which there does not exist a finite variance no less than 0 and no greater than $\tfrac{1}{2}(n^{2} - n)$.
\end{proof} 

\begin{lemma}~\label{lem:partition}
Let \(F_{\kappa}\) be a distribution function for the Kemeny metric of a random variable realised upon \(\overline{\mathbb{R}}\). For each \(\epsilon>0\) there exists a finite partition of the extended real line such that for an orderable sequence \(\infty^{-} \le t_{1} \le \cdots \le t_{k} \le \infty^{+}\) by Lemma~\ref{lem:lsn}, and  there exists \(0 \le j \le k-1\)
\[F_{\kappa}(t_{j+1})^{-} -  F_{\kappa}(t_{j}) \le \epsilon.\]
\end{lemma}
\begin{proof}
Let \(0 < \epsilon\) be given, such that there exists monotone convergence. Allow \(t_{0} = \inf{\overline{\mathbb{R}}}\), for which \(j \ge 0\) we define \(t_{j+1} = \sup\{z: F(z) \le F(t_{j}) + \epsilon\}.\) Then by right continuity, there are a finite sequence of steps for which this definition is discontinuous, and we observe that for our definition of the \(\kappa\) function, this scenario does not occur upon any countable finite population. Thus is defined a transition state of monotonically decreasing distance sequences from the expectation upon \(F\), and thus a finite distance for any finite sample from the compact and totally bounded supremum Kemeny distance \(n^{2}-n\).
\end{proof}


\begin{lemma}
\label{lem:lower_equality}
The total variational distance between any Euclidean linear function and an unbiased Kemeny linear function is lower-bounded by 0.
\end{lemma}
\begin{proof}
Let the sum of all finite collections of random variables from both sub-Gaussian and Gaussian fields be stable, allowing \(Z\) to be a random set of realised variates of length \(n\), and let \(X\) be an independent random variable with distribution function \(F_{\kappa}\) and for which \(N \in U(\mathcal{M})\) is any finite number defining the compact and totally bounded support of \(F_{\kappa}\). It would then follow that the tail probability of \(X\) beyond \(N\) is the chance that \(X \notin U(\mathcal{M})\). If \(F_{\kappa}\) is the Gaussian distribution then, the probability that \(N \ni F\) and thus that \(N\) is not measurable upon \(F_{\kappa}\), is 0: \[e_{G} = \Pr(|X| > 0,\dots,N) = \Pr(X \notin U(\mathcal{M}) = F(N) - \lim_{\varepsilon \to 0^{+}} F(-N-\varepsilon).\]

Now consider the finite Galois field \(\mathcal{G}\) which contains \(G\), and consists of the Beta-Binomial distribution which is also measured upon random variable \(X\). The total variation distance between \(Y = G(Z)\) and \(X = F(Z)\) is the total variation distance between their probability distributions, where \(\mathcal{F}\) is the (probability measure) Borel set sigma-algebra upon \(F\) and \(\mathcal{G}\) is the corresponding set upon \(G\), from which follows: \[\|X-Y\|_{TV} = \sup_{A\subset \mathcal{F}}\left| \Pr(X\in A) - \Pr(Y \in A) \right|.\] The Markov (Lemma~\ref{lem:markov}) and Tchebyshef (Lemma~\ref{lem:chebyshev}) inequalities are immediately seen to hold for any sub-Gaussian function measured upon the Kemeny distance function. This is due to possessing finite expectations (compact and totally bounded, Lemma~\ref{lem:kem_bounded}), as the expectation of \(|X|\) is bounded above by the integral of the common population, which is finite (Definition~\ref{def:stict_sg}) and therefore are always measurable for any random variable which arises. 

Assume \(N\) is large enough to make \(\Pr(X \in U(\mathcal{M}))\) trivially non-zero, satisfying all conditions for which \(n >2\). Truncating \(X\) at a rational fraction of \(N\) thereby removes all chance that it exceeds \(N\), for which the value of the new distribution function at \(G(x)\) is \(0\) for \(x < -N\), \(1\) for \(x \ge N\) as defined upon \(F(x)\), and otherwise equals \[F_{[N]}(x) = \frac{F(x) - \lim_{\varepsilon \to 0^{+}} F(-N-\varepsilon)}{1-e_F(N)}.\] 

Assume instead that \(F(x)\) is a biased function, allowing \(0 < p < 1\) and allow \(\gamma^{2}_{M}\) to be any positive number, representing the amount by which we wish to shift the expectation of the distribution of \(X\) upon the expected asymptotic population defined by the Kemeny metric. Should  \(X\) be any random variable with a finite expectation, then allow \(\varepsilon > 0\). Pick an \(N\) for which \(e_{F}(N) \le \tfrac{\varepsilon}{2}\) and truncate \(X\) at \(N\). It then follows that the mean of \(E(X_{[N]}) = 0 + \gamma^{2}_{M}\) expressed as a \(\frac{\varepsilon}{2}\)-mixture, which changes the total variation distance by at most \(\tfrac{\varepsilon}{2}\). By sub-additivity then follows
\begin{equation}
\|X - Y\|_{TV} \le \|X - X_{[N]}\|_{TV} + \|X_{[N]} - Y\|_{TV} \le \frac{\varepsilon}{2} + \frac{\varepsilon}{2} = \varepsilon,
\end{equation}
as there exists no element \(N \notin U(\mathcal{M})\).

This therefore proves that for any \(X\) exists a finite expectation, and therefore and independently distributed random variable within finite expectation upon the Kemeny metric, there is always a way to truncate \(X\) and shift its mean to \(\gamma^{2}_{M}\), no matter what value \(\gamma^{2}_{M}\) might have, without moving by more than \(\varepsilon\) in the total variation distance. These constructions put an upper bound on \(\|Y\|\): it is no greater than the larger of $|\tfrac{N}{2}|$, representing the absolute value of the position of the atom located at \(2\tfrac{(\gamma^{2}_{M} - E(X))}{\varepsilon}\). In consequence,  the tails of $Y$ are zero, making them sub-Gaussian. As $\varepsilon$ may be arbitrarily small, the only possible lower bound on the distance is zero, and the affine linear invariance of any compact and totally bounded, or complete metric, space allows the estimated parameter to almost surely be upon \(U(\mathcal{M})\), for all finite \(n\).
\end{proof}

\begin{lemma}
\label{lem:upper_equality}
The total variational distance between any Euclidean linear function and an unbiased Kemeny linear distance function is upper-bounded by a finite function of \(\Pr_{\kappa}(X)\), corresponding to a finite distance of \(|\tfrac{n^{2}-n}{2}|\ge{0}.\)
\end{lemma}
\begin{proof}
The lower bound is always finite and may be treated as 0 w.l.g. for any affine linear function upon a metric space, for an arbitrary collection of points. The upper bound for the performance of a system \(\rho(\hat{Y},Y)\), for arbitrary metric \(\rho\) which may be indeterminate (due to the lack of a compact and totally bounded domain upon the extended real line). For the Euclidean metric space, this issue is identified by the use of `approximately correct systems' which bound the measure of the space to be a finite value. With the extended real line \(\overline{\mathbb{R}}\) of performance, for which in conjunction with the Kemeny metric \(\rho_{\kappa}\) we obtain finite moments for arbitrary measure spaces, we show that the total variational distance is almost surely finitely upper-bounded for any homogeneous function space (i.e., an affine linear function space for a common population). 

Allow \(X\) to not contain finite expectations and also be a random variable. Then as \(X\) is unmeasurable, for the moments are not in the real line upon the Euclidean metric space, the probability bounds are almost surely only guaranteed (measurable with probability 1) upon the Kemeny metric space for which the Kemeny distance support is defined \(U(\mathcal{M})\) about 0. For any non-constant vector \(X\) then, all measurable distances between \(X\) and \(Y\) are upon \(U(\mathcal{M})\), with realised error \(\sup |\tfrac{n^{2}-n}{2}|.\) Then the maximum error is almost surely in the neighbourhood about 0, \(\Pr(2|(\tfrac{n^{2}-n}{2})| \in U(\mathcal{M})) = 1\), guaranteeing convergence.
\end{proof}

By the strong law of large numbers then, we ensure that a linear function converges asymptotically to the true distribution, and one which does so stably for all finite \(n\), by the strict sub-Gaussian nature of its distribution (Lemma~\ref{lem:strict_subgauss}). Further, for each \(n\) we may partition the observation of an independently random vector \(x\) within \(\mathcal{M}_{n}\) as a consequence of the orthonormal relationship conditionally observed in isolation (i.e., assuming i.i.d.). Therefore, an unbiased estimator of the CDF of \(F_{\kappa}\) arises, which by Lemma~\ref{lem:lower_equality} and Lemma~\ref{lem:upper_equality} are shown to converge to the true distribution by the weak law of large numbers, regardless of the non-linear score relationship by the Total Variational Distance upon the \(\sigma_{\kappa}\)-algebra of the Kemeny metric space. The distribution however is not self-evident, however by empirical moment matching, we obtained the distribution of the Kemeny distance to be Beta-Binomial in distribution for all finite \(n\). The compound nature of the probability distribution is a conventional problem, however as the support is almost surely compact and totally bounded with regular probability, the Binomial nature of the distribution may be removed for any sample by multiplication of \(U(\mathcal{M}_{n}) \cdot \sup U(\mathcal{M}_{n})\), which cumulatively normalises the support upon \([0,1]\).

\begin{lemma}
\label{lem:cochran}
The Kemeny estimator satisfies Cochran's theorem under uniform and independent sampling.
\end{lemma}
\begin{proof}
Consider the set of \(\kappa^{2}(X)\) values for any domain space \(X \in \mathcal{M}_{n}\), upon which the variance is strictly positive and less than one, and for which the first-order expectation \(E(\kappa(X)) = 0\) for all \(m\) (by Lemma~\ref{lem:unbiased}). By the Chernoff bound then, the known tails of the CDF may be obtained at the truncated points \(a = \frac{n-n^{2}}{2},b = \frac{n^{2}-n}{2}\), which is the support of \(U(\mathcal{M}_{n})\). As the \(\kappa\) function is the linear basis of a complete metric Hilbert space, it is closed under both addition and multiplication, and therefore allows for the sum of \(\chi^{2}\) variables to also be distributed as such. This is accepted due to the linear addition upon the domain \(x,y\) surely possessing a corresponding co-image mapping by equation~\ref{eq:kem_score}, which is invariant under linear and monotone translation; corresponding conditions for this extension were already proven in \textcite{semrl1996}, and therefore we conclude that Cochran's theorem holds for randomly sampled variables of length \(n\) upon the Kemeny measure space, subject to the assertion that all scores are linearly orderable (such as arising from the exponential family of distributions, or any discretisation thereupon). It then directly follows that the \(\chi^{2}_{1}\) distribution holds, by the normality of the approximately linear projective function Hilbert space which enables the acceptance of Chernoff's bound for the discrete uniform distribution without issue, purely as a function of the already proven generalised central limit theorem for strictly sub-Gaussian random variables which are linearly orderable.
\end{proof}

\section{Density of the permutation space \(\mathcal{M}_{n}\)}
\label{lem:density}
 \begin{equation*}
\noindent
 \scriptsize
 \begin{split}
 n^{n} & < {\sqrt {2\pi n}}\ \left({\frac {n}{e}}\right)^{n}e^{\frac {1}{12n+1}}\\
n\log(n) & < \frac{1}{2} \log(2\pi n) + n (\log(n) - \log(e)) + \big(\log(1) - \log(12n +1)\big)\log(e) < n! \\
& \hspace{1cm} <  n\log(n)  < \frac{1}{2} \log(2\pi n) + n (\log(n) - \log(e)) + \big(\log(1) - \log(12n)\big)\log(e)\\
n\log(n) & < \frac{1}{2} \log(2\pi n) + n (\log(n) - 1) + \big(0 - \log(12n +1)\big) < n! \\
& \hspace{1cm}<  \frac{1}{2} \log(2\pi n) + n (\log(n) - 1) + \big(0 - \log(12n)\big)\\
n\log(n) & < \frac{1}{2} \log(2\pi) + n ( \log(n) + \log(n) - 1) \equiv (\log(n)) + \big(0 - \log(12n +1)\big)\\
 n^{n} & < \sqrt{2\pi n} \cdot (\frac{n}{e})^{n}\\
 n\log(n) & < \log(2\pi n) + n\big(\log(n) - \log(e)\big)\\
 n \log(n) & < \frac{1}{2} \log\big(2\pi n\big) + n \log(n) - n\log(e)\\
 0 & < \frac{1}{2}\log\big(2\pi n) - n\\
 2n & < \log\big(2\pi n\big)\\
 \end{split}
 \end{equation*}

We establish the existence of a neighbourhood around the expectation, which is equivalent to establishing the consistency and convergence of any random variables measured upon the abelian linear Kemeny function space. Let \(d_{\kappa} \in [0,n^{2}-n]\) as measured in equation~\ref{eq:kem_dist}. The largest possible Kemeny distance for any space in \(\mathcal{M}_{n}\) is obtained upon the Identity permutation \(I_{n \times 1} = 1,2,\ldots,n-1,n\) and its reverse \(I^{\prime}_{n \times 1} = n, n -1,\ldots,2,1\), whose distance is always \(n^{2}-n\) for all finite \(n\), and denotes a linear distance (as a metric space) over the extended real domain upon \(X,Y\), which is equivalent to the graphical diameter of \(\mathcal{M}_{n}\). \(E(d_{\kappa}(\mathcal{M}) = \tfrac{n^{2}-n}{2},\) and any affine linear transformation, including the subtraction of the expectation upon equation~\ref{eq:kem_dist}, is still a complete Hilbert space, by the definition. Further, all metrisable spaces are recognised to be perfectly normal \(T_{6}\) Haussdorff spaces. By subtracting the set of all distances from about the expectation, the neighbourhood of the Kemeny metric is now signed about 0: \(U(\mathcal{M}_{n}) = [-\tfrac{n^{2}-n}{2},\tfrac{n^{2}-n}{2}]\). Said neighbourhood, which is almost surely finite as the mapping of \(n^{n} - n \mapsto [0,n^{2}-n]\), is an equivalent condition which establishes the consistency of the Kemeny distance. Therefore,
\begin{lemma}
\label{lem:kem_bounded}
The Kemeny metric space is compact and totally bounded for any finite \(n\).
\end{lemma}
\begin{proof}
The Kemeny metric is complete by as a Hilbert space, and totally bounded (Lemma~\ref{lem:kem_bounded}), in that it possesses no point of finite $n$ reals, using finite positive real scalar  \(0 < a \in \mathbb{R}^{1 \times 1}<\infty^{+}\), which is outside the bounds of $0 \le d_{\kappa}(a\kappa^{*}(x_{A}),a\kappa^{*}(x_{B})) \le a(n^{2}-n), \forall\ 0 < a < \infty^{+}$ for any complete space, where \(a\) is an arbitrary finite scalar. Therefore the Kemeny metric space is shown to be both compact and complete and consequently separably dense as well as a \(T_{6}\) topological space.
\end{proof}

\begin{lemma}\label{lem:gc}
The Kemeny metric function over \(\mathcal{M}_{n}\) satisfies the Glivenko-Cantelli theorem: Let \(\{x_{i}\}_{i=1}^{\mathcal{M}}\) be an independently distributed uniform sequence of random variables with distribution function \(F \in \overline{\mathbb{R}}\). Then \[\sup_{x\in\overline{\mathbb{R}}} | \hat{F}_{m}(x) - F(x)| \to 0, a.s.\]
\end{lemma}

\begin{proof}
For any \(\epsilon>0\), holds \[\lim_{m\to\infty}\sup_{x\in\overline{\mathbb{R}}} |\hat{F}_{m}(x) - F(x)| \le \epsilon, a.s.\]  By Lemma~\ref{lem:partition} exists a partition index \(j\) for which \(t_{j} \le x < t_{j+1}\), satisficing:
\begin{equation*}
\begin{aligned}
\hat{F}_{m}(t_{j}) \le \hat{F}_{m}(x) \le \hat{F}_{m}(t^{-}_{j+1}) \land F(t_{j}) \le F(x) \le F(t^{-}_{j+1}),\\
\implies \hat{F}_{m}(t_{j}) - F(t^{-}_{j+1}) \le \hat{F}_{m}(x) - F(x) \le \hat{F}_{m}(t^{-}_{j+1}) - F(t_{j}) \equiv \\
\hat{F}_{m}(t_{j}) - F(t_{j}) + F(t_{j}) - F(t^{-}_{j+1}) \le \hat{F}_{m}(x) - F(x) \le  \hat{F}(t^{-}_{j+1}) - F(t^{-}_{j+1}) + F(t^{-}_{j+1}) - F(t_{j})\\
\therefore \hat{F}_{m}(t_{j}) - F(t_{j}) - \frac{\epsilon}{2} \le \hat{F}_{m}(x) - F(x) \le \hat{F}_{m}(t^{-}_{j+1}) - F(t^{-}_{j+1}) + \frac{\epsilon}{2},
\end{aligned}
\end{equation*}
which tends to equality at 0 by the strong law of large numbers. Thus, the rank ordering of any extended real distribution satisfies the Glivenko-Cantelli theorem upon the Kemeny metric for any finite and therefore countable sample.
\end{proof}

\begin{corollary}
Equation~\ref{eq:kem_dist} is bijectively equivalent to the \textcite{kemeny1959} metric. 
\end{corollary}

\begin{proof}
Consider \(n=1\), for which both distance functions are observed to provide a distance of 0, for any \(x_{i=n} \in \overline{\mathbb{R}}\). For \(n=2,\) upon the \(\mathcal{M}_{n}\) there are \(n^{n}-n = 2\) unique permutations with repetitions allowed, while excluding the \(n\) degenerate random variables. These two permutations are, respectively, the Identity permutation and the reverse Identity permutation, the set \(\{I_{n} = [1,2], I_{n}^{\prime} = [2,1],\) and \(S_{n=2} = \mathcal{M}_{n=2};\) therefore the Kendall \(\tau\)-distance, the Kemeny distance, and equation~\ref{eq:kem_dist} must all possess the same distances, which by symmetry are equivalent. For \(\tau(I_{n},I^{\prime}_{n}),\) we observe the distance \(\tfrac{n^{2}-n}{2} = 1,\) while for the Kemeny distance, we observe the distance of 2. Note however that the set of distances of 1 upon the Kemeny metric are non-measurable upon the Kendall \(\tau\)-distance, and thus would present with a support only upon the even set of numbers. Re-scaling the Kendall distance to be equivalent to the Kemeny distance is obtained via an affine-linear transformation by scalar 2. Then, all distances are validly observed to possess the same distance, and thus the Kendall and Kemeny distances are affine-linear rescalings for \(n=1,2.\) However, it is observed that for any \(n\) then, there is a corresponding finite \(n+1\) which is also finitely defined upon the set of all integers, and thus for \(n = 1,2,3,\ldots,\mathbb{N}^{+},\) the entire set of all samples is measurable with finite distance.  

Consider now equation~\ref{eq:kem_dist} and the corresponding Kemeny distance. Our function possesses an image of \([0,\ldots, n^{2}-n]\) elements, as does the Kemeny metric. However, upon \(n=1,\) the only observable distance is between a random vector and itself, and thus always has a distance of 0. For \(n=2\), we have observed that there are two unique permutations: \(\{I_{n},I_{n}^{\prime}\},\) each of which have distance of 2 from itself. In both instances, we observe that both the domain and the images are respectively identical. Finally, consider the finite domain \(\mathcal{M}_{n} < \infty^{+},\) which is true for any \(n \in \mathbb{N}^{+},\) upon which is observed a distance no less than one and no greater than \(n^{2}-n.\) Affine-linear rescaling of said distances, to set the expectation of the population of distances \(E(\mathcal{M}_{n}) = 0\), does not change the bijective relationship between the two distance metrics. Therefore, the Kemeny metric and that of equation~\ref{eq:kem_dist} are equivalent, and are linear rescalings of the Kendall \(\tau\)-distance when restricted to \(S_{n}\).
\end{proof}

\begin{corollary}
 \label{cor:density}
 Assume that for $n>0$, the density of the permutation spaces for the respective measures are $\mathcal{M} = n^{n} - n$ for the Kemeny metric space, and $\mathcal{M}^{\prime} = {n}!$ for the Kendall metric space, where $\mathcal{M}^{\prime} \subset \mathcal{M}$. It would then follow that the number $\mathcal{M}$ must be well represented for the Stirling approximation of factorials \[{\sqrt {2\pi n}}\ \left(\frac {n}{e}\right)^{n}e^{\frac {1}{12n+1}} < {n}! < {\sqrt {2\pi n}}\ \left({\frac {n}{e}}\right)^{n}e^{\frac{1}{12n}},\] or else \(\mathcal{M}^{\prime} \subset \mathcal{M}_{n}\).\end{corollary}
 \begin{proof}
 Assume $\mathcal{M} \subseteq \mathcal{M}^{\prime}$, denoting that the $\tau$-distance space is of greater or equal density to the Kemeny $\rho_{\kappa}$-space. If this were so, it would therefore follow that \(2n \le \log\big(2\pi n\big)\), which is strictly false by contradiction under all conditions, as there is no $n\in \mathbb{N}^{+}$ which satisfies this strict inequality. It follows then that $\mathcal{M}^{\prime} \subset \mathcal{M}$, and therefore that all measurements upon the Kendall distance are a strict subset within the Kemeny distance, $\tau(x,y) \subset \rho_{\kappa}(x,y)$. As the permutations with ties must always strictly subsume the set of all permutations without ties, the set difference must always be positive \(\lim_{n\to\infty^{+}} n^{n}-n - n! > 0,\) a relational which holds uniformly for all \(n>2\), and thereby ensures that the density of the set-space of all permutations with ties is almost surely greater.
 \end{proof}

\begin{corollary}
If the Kemeny distance is almost surely strictly sub-Gaussian, we observe that the distribution is centred at 0, and possesses symmetric tails of density which is less than or equal to that of a standard normal distribution. Therefore it immediately follows that four moments are sufficient to characterise the probability distribution upon any population of size \(\mathcal{M}_{n}\).
\end{corollary}
\begin{proof}
It is self-evident that any power of an expectation of 0 is also 0, and therefore that all higher order odd-moments are equal to 0 for the Kemeny metric upon \(U(\mathcal{M}_{n})\). The second central moment is defined in equation~\ref{eq:kem_variance} and is therefore orthonormal of the expectation of 0; by the symmetry of the distance (Lemma~\ref{lem:even}) the skewness is 0, and therefore the final free moment to examine is the excess kurtosis \(\mu_{4}\) upon a finite sample. The Kemeny distribution however is strictly sub-Gaussian and therefore must posses negative excess kurtosis for any finite \(n\), as otherwise it would be normally distributed for finite \(n\) or fail to satisfy the Borel-Cantelli lemma. This unique probability distribution is therefore symmetric and unbiased, has a spread and measure of all scores which is almost surely positive and finite, possesses no skewness, and finite kurtosis which tends to 3 asymptotically from below \parencite[Ch.~1]{buldygin2000}. This paradoxically contradicts the asserted conclusion that the asymptotic performance stably holds for finite samples (i.e., is stable), as the distribution for the finite sample (strictly sub-Gaussian) distances cannot be normally distributed (sub-Gaussian) in the presence of ties upon finite samples. Notice then that the unbiased expectation, by the Lemma~\ref{lem:clt_kem}, is an unbiased linear function about the median which must converge by the strong law of large numbers to the mean in the asymptotic limit, with finite variance.
\end{proof}

\end{document}